\documentclass[11pt]{article}

\usepackage{mdframed}

\def\showauthornotes{1}
\def\showtableofcontents{0}
\def\showkeys{0}
\def\showdraftbox{0}
\def\showcolorlinks{1}
\def\usemicrotype{1}
\def\showfixme{0}

\usepackage[boxed]{algorithm2e}

\usepackage[dvipsnames]{xcolor}
\usepackage{tikz}
\usepackage{etex}

\usepackage[l2tabu, orthodox]{nag}

\usepackage[utf8]{inputenc}

\usepackage{xspace,enumerate}

\usepackage[T1]{fontenc}
\usepackage[full]{textcomp}

\usepackage[american]{babel}

\usepackage{mathtools}

\usepackage[normalem]{ulem}

\usepackage{bm}

\usepackage{amsthm}

\newtheorem{theorem}{Theorem}[section]
\newtheorem*{theorem*}{Theorem}

\newtheorem*{proposition*}{Proposition}
\newtheorem{lemma}[theorem]{Lemma}
\newtheorem*{lemma*}{Lemma}

\newtheorem*{conjecture*}{Conjecture}
\newtheorem{fact}[theorem]{Fact}
\newtheorem*{fact*}{Fact}

\newtheorem*{hypothesis*}{Hypothesis}

\theoremstyle{definition}
\newtheorem{definition}[theorem]{Definition}

\theoremstyle{remark}
\newtheorem{claim}[theorem]{Claim}
\newtheorem*{claim*}{Claim}
\newtheorem{remark}[theorem]{Remark}
\newtheorem*{remark*}{Remark}

\newtheorem*{observation*}{Observation}

\usepackage[letterpaper,
 top=1in,
 bottom=1in,
 left=1in,
 right=1in]{geometry}

\usepackage[varg]{pxfonts} % varg - uses nicer g,v,y,w

\ifnum\showkeys=1
\usepackage[color]{showkeys}
\fi

\ifnum\showcolorlinks=1
\usepackage[
pagebackref,
letterpaper=true,
colorlinks=true,
urlcolor=blue,
linkcolor=blue,
citecolor=OliveGreen,
]{hyperref}
\fi

\ifnum\showcolorlinks=0
\usepackage[
pagebackref,
letterpaper=true,
colorlinks=false,
pdfborder={0 0 0}
]{hyperref}
\fi

\usepackage{prettyref}

\newcommand{\savehyperref}[2]{\texorpdfstring{\hyperref[#1]{#2}}{#2}}

\newrefformat{eq}{\savehyperref{#1}{\textup{(\ref*{#1})}}}
\newrefformat{lem}{\savehyperref{#1}{Lemma~\ref*{#1}}}
\newrefformat{def}{\savehyperref{#1}{Definition~\ref*{#1}}}
\newrefformat{thm}{\savehyperref{#1}{Theorem~\ref*{#1}}}
\newrefformat{cor}{\savehyperref{#1}{Corollary~\ref*{#1}}}
\newrefformat{cha}{\savehyperref{#1}{Chapter~\ref*{#1}}}
\newrefformat{sec}{\savehyperref{#1}{Section~\ref*{#1}}}
\newrefformat{app}{\savehyperref{#1}{Appendix~\ref*{#1}}}
\newrefformat{tab}{\savehyperref{#1}{Table~\ref*{#1}}}
\newrefformat{fig}{\savehyperref{#1}{Figure~\ref*{#1}}}
\newrefformat{hyp}{\savehyperref{#1}{Hypothesis~\ref*{#1}}}
\newrefformat{alg}{\savehyperref{#1}{Algorithm~\ref*{#1}}}
\newrefformat{rem}{\savehyperref{#1}{Remark~\ref*{#1}}}
\newrefformat{item}{\savehyperref{#1}{Item~\ref*{#1}}}
\newrefformat{step}{\savehyperref{#1}{step~\ref*{#1}}}
\newrefformat{conj}{\savehyperref{#1}{Conjecture~\ref*{#1}}}
\newrefformat{fact}{\savehyperref{#1}{Fact~\ref*{#1}}}
\newrefformat{prop}{\savehyperref{#1}{Proposition~\ref*{#1}}}
\newrefformat{prob}{\savehyperref{#1}{Problem~\ref*{#1}}}
\newrefformat{claim}{\savehyperref{#1}{Claim~\ref*{#1}}}
\newrefformat{relax}{\savehyperref{#1}{Relaxation~\ref*{#1}}}
\newrefformat{red}{\savehyperref{#1}{Reduction~\ref*{#1}}}
\newrefformat{part}{\savehyperref{#1}{Part~\ref*{#1}}}

\newcommand{\Sref}[1]{\hyperref[#1]{\S\ref*{#1}}}

\usepackage{nicefrac}
\ifnum\usemicrotype=1
\usepackage{microtype}
\fi

\ifnum\showauthornotes=1
\newcommand{\Authornote}[2]{{\sffamily\small\color{red}{[#1: #2]}}}
\newcommand{\Authornotecolored}[3]{{\sffamily\small\color{#1}{[#2: #3]}}}
\newcommand{\Authorcomment}[2]{{\sffamily\small\color{gray}{[#1: #2]}}}
\newcommand{\Authorstartcomment}[1]{\sffamily\small\color{gray}[#1: }

\newcommand{\Authorfnote}[2]{\footnote{\color{red}{#1: #2}}}
\newcommand{\Authorfixme}[1]{\Authornote{#1}{\textbf{??}}}
\newcommand{\Authormarginmark}[1]{\marginpar{\textcolor{red}{\fbox{\Large #1:!}}}}
\else
\newcommand{\Authornote}[2]{}
\newcommand{\Authornotecolored}[3]{}
\newcommand{\Authorcomment}[2]{}
\newcommand{\Authorstartcomment}[1]{}

\newcommand{\Authorfnote}[2]{}
\newcommand{\Authorfixme}[1]{}
\newcommand{\Authormarginmark}[1]{}
\fi

\newcommand{\Cnote}{\Authornote{C}}

\ifnum\showfixme=0

\fi

\usepackage{boxedminipage}

\usepackage{mathrsfs}

\newcommand{\textparen}[1]{\text{(#1)}}

\ifx\because\undefined
\newcommand{\because}[1]{\textparen{because #1}}
\else
\renewcommand{\because}[1]{\textparen{because #1}}
\fi

\newcommand\bdot\bullet

\DeclareMathOperator{\OPT}{OPT}

\newcommand{\problemmacro}[1]{\texorpdfstring{\textsc{#1}}{#1}\xspace}

\newcommand{\cC}{\mathcal C}

\newcommand{\cH}{\mathcal H}

\newcommand{\cL}{\mathcal L}

\newcommand{\cN}{\mathcal N}

\newcommand{\cP}{\mathcal P}

\newcommand{\cR}{\mathcal R}

\renewcommand{\leq}{\leqslant}

\renewcommand{\geq}{\geqslant}

\ifnum\showdraftbox=1

\else

\fi

\let\epsilon=\varepsilon

\numberwithin{equation}{section}

\newcommand{\MYstore}[2]{%
  \global\expandafter \def \csname MYMEMORY #1 \endcsname{#2}%
}

\newcommand{\MYload}[1]{%
  \csname MYMEMORY #1 \endcsname%
}

\newcommand{\MYnewlabel}[1]{%
  \newcommand\MYcurrentlabel{#1}%
  \MYoldlabel{#1}%
}

\newcommand{\MYdummylabel}[1]{}

\newcommand{\torestate}[1]{%
  \let\MYoldlabel\label%
  \let\label\MYnewlabel%
  #1%
  \MYstore{\MYcurrentlabel}{#1}%
  \let\label\MYoldlabel%
}

\newcommand{\restatetheorem}[1]{%
  \let\MYoldlabel\label
  \let\label\MYdummylabel
  \begin{theorem*}[Restatement of \prettyref{#1}]
    \MYload{#1}
  \end{theorem*}
  \let\label\MYoldlabel
}

\newcommand{\restatelemma}[1]{%
  \let\MYoldlabel\label
  \let\label\MYdummylabel
  \begin{lemma*}[Restatement of \prettyref{#1}]
    \MYload{#1}
  \end{lemma*}
  \let\label\MYoldlabel
}

\newcommand{\restateprop}[1]{%
  \let\MYoldlabel\label
  \let\label\MYdummylabel
  \begin{proposition*}[Restatement of \prettyref{#1}]
    \MYload{#1}
  \end{proposition*}
  \let\label\MYoldlabel
}

\newcommand{\restatefact}[1]{%
  \let\MYoldlabel\label
  \let\label\MYdummylabel
  \begin{fact*}[Restatement of \prettyref{#1}]
    \MYload{#1}
  \end{fact*}
  \let\label\MYoldlabel
}

\newcommand{\restate}[1]{%
  \let\MYoldlabel\label
  \let\label\MYdummylabel
  \MYload{#1}
  \let\label\MYoldlabel
}

\let\origparagraph\paragraph
\renewcommand{\paragraph}[1]{\origparagraph{#1.}}

\allowdisplaybreaks

\usepackage{enumitem}
\newcommand{\DP}[3]{\ensuremath{\DPm_{#1}(#2,#3)}}
\newcommand{\FLOW}[3]{\ensuremath{\FLOWm_{#1}(#2,#3)}}
\newcommand{\conf}[1]{\ensuremath{\mathcal{C}(#1)}}
\newcommand{\res}{\ensuremath{\mathcal{R}}\xspace}
\newcommand{\players}{\ensuremath{\mathcal{P}}\xspace}

\DeclareMathOperator{\DPm}{DP}
\DeclareMathOperator{\FLOWm}{F}

\title{\bfseries Combinatorial Algorithm for Restricted Max-Min Fair Allocation}

\author{
Chidambaram Annamalai\thanks{Department of Computer Science,
    ETH Zurich. Email:
\href{mailto:cannamalai@inf.ethz.ch}{cannamalai@inf.ethz.ch}. \newline Work performed while the author was at the School of Basic Sciences, EPFL.},
Christos Kalaitzis\thanks{School of Computer and Communication Sciences, EPFL. 
Email:
\href{mailto:christos.kalaitzis@epfl.ch}{christos.kalaitzis@epfl.ch}.  Supported by ERC Starting Grant 335288-OptApprox.},
Ola Svensson\thanks{School of Computer and Communication Sciences, EPFL.
Email:
\href{mailto:ola.svensson@epfl.ch}{ola.svensson@epfl.ch}. Supported by ERC Starting Grant 335288-OptApprox.}
}

\begin{document}

\maketitle
%\draftbox
\thispagestyle{empty}

\begin{abstract}
We study the basic allocation problem of assigning  resources to players so as to maximize
fairness. % This is the problem of \problemmacro{restricted max-min fair allocation}.
% i.e., we wish to maximize the happiness of the least happy player. Each
%tem has value $p_j$ and can be assigned to a subset of the players.
This is one of the few natural problems that enjoys the intriguing status of having a better estimation
algorithm than approximation algorithm. Indeed, %Asadpour, Feige, and Saberi showed that 
 a certain
Configuration-LP can be used to estimate the value of the optimal allocation to within a factor of $4+\epsilon$.
In contrast, however, the best known approximation algorithm for the problem has an unspecified large constant guarantee.

% On the other hand,  the previously best  approximation algorithm   first preprocesses the
%solution to the configuration-LP and then repeatedly applies Lovasz Local Lemma. Each step incurs a
%loss in the approximation guarantee and the final 

%an unspecified small
%constant guarantee. 

%has an unspecified approximation
%guarantee in the millions.

% Apart from being an interesting
%problem of its own, it is a rare instance of a natural problem where we have stronger estimation
%algorithms, a polynomial time algorithm that estimates the optimal value, than approximation
%algorithms, a polynomial time algorithm that finds a solution with a guaranteed value compared to
%the optimal value. 
%for estimating the optimal value than actually finding a solution with a guaranteed value.

In this paper we significantly narrow this gap by giving a $13$-approximation
algorithm for the problem. 
%Moreover, our algorithm is combinatorial and only uses the configuration-LP in the
%analysis. This is 
%
Our approach develops a local search technique introduced by Haxell~\cite{haxell1995condition} for
hypergraph matchings, and later used in this context by
Asadpour, Feige, and Saberi~\cite{asadpour2012santa}. For our local search procedure to terminate in polynomial time, we introduce several
new ideas such as \emph{lazy updates} and \emph{greedy players}. Besides the improved approximation guarantee, 
the highlight of our approach %compared to previous approximation algorithms
is that it is purely combinatorial and uses the Configuration-LP only in the analysis. 

%To make the 
% Our
%approach 

%also differs from previous approximation algorithms in that it does not rely on solving the
%large configuration-LP. Instead we use a combinatorial local search and only rely on the relaxation
%for the analysis.

%improves upon the previous best algorithm which achieved an unspecified very large constant
%guarantee. Our methods are also very different. The previous approximation algorithm worked by first
%solving a strong LP-relaxation, called the configuration-LP, followed by preprocessing and then
%repeatedly applying Lovasz local lemma. Our algorithm on the other hand is 

%%% Local Variables:
%%% mode: latex
%%% TeX-master: "main"
%%% End: 

\end{abstract}


\clearpage

% tableofcontents added for better navigability of the document
\ifnum\showtableofcontents=1
{
\tableofcontents
\thispagestyle{empty}
 }
\fi

\clearpage

\setcounter{page}{1}

\section{Introduction}
We consider the \problemmacro{Max-Min Fair Allocation} problem, a basic combinatorial optimization
problem, that captures the dilemma of how to allocate resources to players in a fair manner.
A problem instance is defined by a set $\cR$ of indivisible resources, a set $\cP$ of players, and a set of nonnegative values $\{v_{ij}\}_{i\in\cP, j\in\cR}$
 where each player $i$ has a value $v_{ij}$ for a resource $j$.
%% and the value of a set $R \subseteq \cR$ of resources is additive, that is, for player $i$ the
%% value of $R$ is $\sum_{j\in R} v_{ij}$.
An \emph{allocation} is simply a partition $\{R_i\}_{i \in \cP}$ of
the resource set and the valuation function
$v_i:2^{\mathcal{R}}\mapsto \mathbb{R}$ for any player $i$ is
additive, i.e., $v_i(R_i) = \sum_{j \in R_i} v_{ij}.$ Perhaps the most
natural fairness criterion in this setting is the max-min objective
which tries to find an allocation that maximizes the minimum value of
resources received by any player in the allocation.  Thus, the goal in this problem is to find an allocation $\{R_i\}_{i\in\cP}$ that maximizes
$$
\min_{i\in \cP} \sum_{j\in R_i} v_{ij}.
$$
%% no ``the'' before santa claus problem
This problem has also been given the name \problemmacro{Santa Claus problem} as interpreting the players as kids and the resources as presents leads to Santa's annual allocation problem of making the least happy kid as happy as possible.

A closely related problem is the classic scheduling problem of
\problemmacro{Scheduling on Unrelated Parallel Machines to Minimize Makespan}. That problem has
the same input as above and the only difference is the objective function: instead of
maximizing the minimum we wish to minimize the maximum. In the scheduling context, this corresponds
to minimizing the time at which all jobs (resources) are completed by the machines (players) they were scheduled on. In a
seminal paper, Lenstra, Shmoys, and Tardos~\cite{lenstra1990approximation} showed that the scheduling problem admits a
$2$-approximation algorithm by rounding a certain linear
programming relaxation often referred to as the Assignment-LP.
%They also showed that the problem is NP-hard to approximate within a
%factor less than $3/2$.  
Their approximation algorithm  in fact has the often
 stronger guarantee
that the returned  solution has value
at most $\OPT + v_{\max}$, where 
% is obtained by solving a certain linear
%program, often referred to as assignment-LP, and then round the fractional solution to obtain an
%integral solution of value at most the value $\tau$ of the fractional solution plus the size of the
$v_{\max} := \max_{i\in \cP, j\in \cR} v_{ij}$ is the maximum value of a job (resource).
%\footnote{To be completely
  %accurate a preprocessing is needed that removes the potential assignments of items to players that
  %are larger than the LP value $\tau$ so $v_{max} := \max_{i \in \cP, j\in \cR: v_{ij} \leq \tau}
  %v_{ij}$.} 

From the similarity between the two problems, it is natural to expect that 
the techniques developed for the scheduling problem are also applicable in this context. What is
perhaps surprising is that the guarantees have not carried over so
far, contrary to expectation. While a rounding of the
Assignment-LP has been shown~\cite{bezakova2005allocating} to provide an allocation of value at
least $\OPT-v_{\max}$, this guarantee deteriorates with increasing $v_{\max}$. Since in hard
instances of the problem (when $v_{\max} \approx OPT$) there can be players who are assigned only one resource in an optimal allocation, this result provides no guarantee in general. %%  there exist resources which, in an optimal allocation, are assigned
%%  Bezakova and Dani~\cite{} formalized this intuition to give a
%% rounding of the assignment-LP with a guaranteed value of 
%% $\OPT - v_{\max}$ when we maximize the minimum.
%% This result gives a good guarantee when the optimal
%% value is large compared to the value of the largest item. However, in contrast to the scheduling
%% problem, it gives no guarantee for the hard instances when $v_{\max} \approx \OPT$.
The lack of
guarantee is in fact intrinsic
%Indeed, it is
%rather intuitive that it is easier to fairly allocate a large set of items of small values than when
%some of the indivisible items are of very large value. The lack of any guarantee is in fact intrinsic
to the Assignment-LP for \problemmacro{Max-Min Fair Allocation} as the
relaxation is quite weak. It has an
unbounded integrality gap i.e., the optimal value of the linear program can be
a polynomial factor larger than the optimal value of an integral solution.

To overcome the limitations of the Assignment-LP, Bansal and Sviridenko~\cite{bansal2006santa} proposed to use a 
stronger relaxation, called Configuration-LP, for \problemmacro{Max-Min Fair Allocation}. Their paper
contains several results on the strength of the Configuration-LP, one negative and many
positive. The negative result says that even the stronger Configuration-LP has 
 % a polynomial 
an integrality gap that grows as $\Omega(\sqrt{|\cP|})$. Their positive results apply for the interesting case when $v_{ij} \in \{0, v_j\}$, called
\problemmacro{Restricted Max-Min
  Fair Allocation}. For this case they give an $O(\log \log |\cP|/ \log \log \log |\cP|)$-approximation
algorithm, a substantial improvement over the integrality gap of the Assignment-LP. Notice that the
restricted version has the following natural interpretation: each resource $j$ has a
fixed value $v_j$
but it is interesting only for some subset of the players. 

Bansal and Sviridenko further showed that the solution to a certain combinatorial problem on set
systems would imply a constant integrality gap. This was later settled positively by
Feige~\cite{feige2008allocations} using a proof technique that repeatedly used the Lov\'asz Local
Lemma. At the time of Feige's result, however, it was not known if his arguments were constructive,
i.e., if it led to a polynomial time algorithm for finding a solution with the same guarantee.
This was later shown to be the case by Haeupler et al.~\cite{haeupler2011new}, who constructivized the
various applications of the Lov\'asz Local Lemma in the paper by Feige~\cite{feige2008allocations}.
This led to the first constant factor approximation algorithm for \problemmacro{Restricted Max-Min
  Fair Allocation}, albeit with a large and unspecified constant.
% \footnote{\Cnote{need to estimate    the previous constant}}. 
%Besides spanning the works of
%~\cite{bansal2006santa,feige2008allocations,haeupler2011new}, 
This approach also requires the solution of the exponentially large Configuration-LP obtained by using the ellipsoid algorithm.

A different viewpoint and rounding approach for the problem was initiated by Asadpour, Feige, and
Saberi~\cite{asadpour2012santa}. This approach uses the perspective of
hypergraph matchings where one can naturally interpret the problem as a
bipartite hypergraph matching problem with bipartitions $\cP$ and $\cR$. Indeed, in a
solution of value $\tau$, each player $i$ is matched to a subset $R_i$ of resources of total value
at least $\tau$ which corresponds to a hyperedge $(i,R_i)$. Previously, Haxell~\cite{haxell1995condition} provided sufficient
conditions for bipartite hypergraphs to admit a perfect matching,
generalizing the well known graph analog, viz., Hall's theorem. Her proof is
algorithmic in the sense that when the sufficient conditions hold, then a perfect matching can be
found using a local search procedure that will terminate after at most exponentially many iterations. Haxell's techniques were successfully adapted by Asadpour et
al.~\cite{asadpour2012santa} to the \problemmacro{Restricted
  Max-Min Fair Allocation} problem to obtain a beautiful proof showing that the Configuration-LP has an integrality gap of at most $4$. As the Configuration-LP can be solved to any
desired accuracy in polynomial time, this gives a polynomial time algorithm to estimate the
value of an optimal allocation up to a factor of $4+\epsilon$, for any
$\epsilon >0$. Tantalizingly, however, the techniques
of~\cite{asadpour2012santa} do not yield an efficient algorithm for
finding an allocation with the same guarantee.

The above results lend the \problemmacro{Restricted Max-Min Fair Allocation} problem an intriguing
status that few other natural problems enjoy (see~\cite{feige2008estimation} for a comprehensive discussion on the
difference between estimation and approximation algorithms). Another problem with a similar status is the
restricted version of the aforementioned scheduling problem. The techniques in~\cite{asadpour2012santa} inspired
the last author to show~\cite{svensson2012santa} that the Configuration-LP estimates the optimal value within a factor
$33/17+\epsilon$ improving on the factor of $2$ by Lenstra et al.~\cite{lenstra1990approximation}. Again, the algorithm in~\cite{svensson2012santa} is
not known to terminate in polynomial time. We believe that this situation illustrates the need for new tools that improve our understanding of the Configuration-LP especially in the context of basic allocation problems in combinatorial optimization.

\paragraph{Our results}
Our main result improves the approximation guarantee for the \problemmacro{Restricted 
 Max-Min Fair Allocation} problem. Note that $6+2\sqrt{10} \approx 12.3$.

\begin{theorem}\label{thm:main}
  %% upper bounding \epsilon by 1 is not necessary -- we can always take \epsilon = min(1, \epsilon') ..
  For every $ \epsilon>  0 $, there exists a combinatorial $(6+2\sqrt{10}+\epsilon)$-approximation
  algorithm for the \problemmacro{Restricted
  Max-Min Fair Allocation} problem that
  runs in time $n^{O(1/\epsilon^2 \log(1/\epsilon))}$ where $n$ is the size of the instance.

%, i.e.,  the number of players and resources and the encoding length of v_{ij} s because we use binary search
\end{theorem}
Our algorithm has the advantage of being completely combinatorial. It does not solve
the exponentially large Configuration-LP. Instead, we use it only in
the analysis to compare the value of the allocation returned by our
algorithm against the optimum. As  our hidden
constants are small, we believe that our algorithm is more attractive
than solving the Configuration-LP for a moderate $\epsilon$. Our approach is
based on the local search procedure introduced in this context by Asadpour et
al.~\cite{asadpour2012santa}, who in turn were inspired by the work of
Haxell~\cite{haxell1995condition}. Asadpour et al. raised the natural question if local
search procedures based on alternating trees can be made to run in polynomial time. 
% Our  work positively answers this question  for our
%modified local search.
%
Prior to this work, the best running time guarantee was a
quasi-polynomial time alternating tree algorithm by Pol\'{a}\v{c}ek and Svensson~\cite{polacek2012quasi}. The main
idea in that paper was to show that the local
% , which is a generalization of the Hungarian
% method for bipartite matchings,
 search can be restricted to alternating paths of length $O(\log n)$ (according
to a carefully chosen length function), where $n$ is the number of players and resources. This restricts the search space of the local search 
giving the running time of $n^{O(\log n)}$. To further reduce the search space seems highly
non-trivial and it is not where our improvement comes from. Rather, in contrast to the previous
local search algorithms, we do not update the partial matching as soon as an alternating path is
found. Instead, we wait until we are guaranteed a significant number of alternating paths, which then intuitively
guarantees large progress. We refer to this concept as \emph{lazy updates}. At the same time, we
ensure that our alternating paths are short by introducing 
\emph{greedy players} into our alternating tree: a player may claim more resources than she needs in an approximate solution.  

To best illustrate these ideas we have chosen to first present a simpler
algorithm in Section~\ref{sec:simplealgo2}. The result of that section still gives an improved approximation guarantee and a polynomial time local search algorithm. However, it is not
combinatorial as it relies on a preprocessing step which in turn uses the solution of the Configuration-LP. Our combinatorial algorithm 
%with also a further improved approximation guarantee  
is then presented in
Section~\ref{sec:combinatorialalgorithm}.  The virtue of explaining the simpler algorithm first is that it allows us to postpone some of the complexities of the combinatorial algorithm until later, while still demonstrating the key ideas mentioned above.

\paragraph{Further related work}
As mentioned before, the Configuration-LP has an integrality gap of $\Omega\left(\sqrt{|\cP|}\right)$ for the
general \problemmacro{Max-Min Fair Allocation} problem.  Asadpour and Saberi~\cite{AS07} almost matched
this bound by giving a $O(\sqrt{|\cP|} \log^3(|\cP|))$-approximation algorithm; later improved by
Saha and Srinivisan~\cite{SahaS10} to $O(\sqrt{|\cP|\log  |\cP|}/\log\log
|\cP|)$. The current best approximation is $O(n^\epsilon)$ due to Bateni et al.~\cite{Bateni09} and  Chakraborty et al.~\cite{Chakrabarty09}; for any $\epsilon>0$ their algorithms run in time
$O(n^{1/\epsilon})$. This leaves a large gap in the approximation guarantee for the general version of the problem as the only known
hardness result says that it is NP-hard to approximate the problem to within a factor less than
$2$~\cite{bezakova2005allocating}. The same hardness also holds for the restricted version.
%\Cnote{fill in later}

%% \Onote{Should we mention graph-balancing results?}
%% \Cnote{No. I think then there is too much related work}

%% \newpage
%% \subsection{Our results and techniques}

%% \begin{theorem}\label{thm:main}
%%   For every $0 < \epsilon < 1$, there exists a combinatorial $(6+2\sqrt{10}+\epsilon)$-approximation algorithm for the restricted max-min fair allocation problem that
%%   runs in time $n^{O(1/\epsilon^2 \log(1/\epsilon))}$ where $n$ is the size of the instance, i.e.,
%%   the number of players and resources.
%% \end{theorem}

\iffalse
\subsection{Organization}
To best illustrate our techniques we have chosen to first present a simpler algorithm that uses a reduction of Bansal and Sviridenko~\cite{bansal2006santa}. This reduction groups players into clusters so that all but one of the players in each group can be allocated fat items. The clustering step allows us to postpone the problem of fat item assignments and focus on demonstrating the key ideas that lead to a local search procedure that terminates in polynomial time. This simpler algorithm already constitutes an improvement over the previous best approximation algorithm. It also motivates our combinatorial algorithm in Section~\ref{sec:combinatorialalgorithm}.

\Cnote{fill in organization at the end}
\fi
%%% Local Variables:
%%% mode: latex
%%% TeX-master: "main"
%%% End: 

\section{The Configuration-LP}\label{sec:CLP}
Recall that a solution to the  \problemmacro{Max-Min Fair Allocation} problem of value $\tau$ is a
partition $\{R_i\}_{i \in \cP}$ of the set of resources so that each player receives a set of
value at least $\tau$, i.e., $v_i(R_i) \geq \tau$ for $i\in \cP$. Let $\conf{i,\tau} = \{C \subseteq \cR: v_i(C) \geq \tau\}$ be the set of
configurations that player $i$ can be allocated in a solution of value $\tau$.
 The
 Configuration-LP has a decision variable $x_{iC}$ for each player $i\in \cP$ and each $C \in \conf{i,\tau}$. The intuition is that
the variable $x_{iC}$
takes value $1$ if and only if she is assigned the bundle $C$. The Configuration-LP is now a
feasibility linear program with two sets of constraints: the first set says that each player should receive
(at least)
one configuration and the second set says that each item should be assigned to at most one player. The
formal definition is given in the left box of Figure~\ref{fig:lp}.

%The configuration linear program (CLP) for the Santa Claus problem is the strongest known convex relaxation for it. 
%The intuition of the CLP is that any
%allocation of value $\tau$ needs to allocate a configuration
%$C$ of resources to each player $i$ so that $v_i(C) \geq \tau$. For a guessed value of $\tau$, the
%configuration LP therefore has a decision variable $x_{i, C}$ for each
%player $i\in \players$ and configuration $C \in \conf{i,\tau}$ with the
%intuition that this variable should take value one if and only if the
%corresponding set of resources is allocated to that player. The configuration LP
%$CLP(\tau)$ is a feasibility program and it is defined as follows:
% \begin{figure}
\begin{figure}[t]
\begin{equation*}
\boxed{%
%  \addtolength{\linewidth}{-2\fboxsep}%
%  \addtolength{\linewidth}{-2\fboxrule}%
       \begin{minipage}{8cm}%
          \begin{align*}
            \sum_{C\in \conf{i,\tau}} x_{iC} &\geq 1,  &  \forall \; i \in \players, \\[0mm]
            \sum_{i,C: j\in C, C\in \conf{i,\tau}} x_{iC} &\leq 1, & \forall \; j \in \res, \\[0mm]
            x & \geq 0.
          \end{align*}
          \vspace{-0.09cm}
        \end{minipage}%
}
\boxed{%
%  \addtolength{\linewidth}{-2\fboxsep}%
%  \addtolength{\linewidth}{-2\fboxrule}%
        \begin{minipage}{8cm}%
          \begin{align*}
            \max & \mbox{  } \sum_{i\in \players} y_i - \sum_{j\in \res} z_j  && \\[2mm]
            y_i &\leq             \sum_{j\in C} z_j, &  \forall i\in \players, \forall C \in
            \conf{i, \tau},\\[1mm]
            y,z & \geq 0.
          \end{align*}
          \vspace{-0.4cm}
%\begin{center}
%Dual of $CLP(\tau)$
%\end{center}
        \end{minipage}%
}
%}
\end{equation*}
 \caption{The Configuration-LP for a guessed optimal value $\tau$ on the left and its dual on the
   right.}
\label{fig:lp}
\end{figure}
% \end{figure}
%The first set of constraints ensures that each player should be assigned
%at least one configuration fractionally and the second set of constraints ensures that a
%resource is not over-assigned.
It is easy to see that if $CLP(\tau_0)$ is feasible, then so is $CLP(\tau)$ for all
$\tau\leq \tau_0$.
%%  because $\conf{i, \tau_0}\subseteq \conf{i, \tau}$ and thus a solution to
%% $CLP(\tau_0)$ is a solution to $CLP(\tau)$ as well.
 We say that the value of the Configuration-LP is $\tau_{OPT}$
if it is the largest value such that the above program is feasible. Since every
feasible allocation is a feasible solution of the Configuration-LP,
$\tau_{OPT}$ is an upper bound on the value of the optimal allocation and therefore $CLP(\tau)$ constitues a valid relaxation.

We note that the LP has exponentially many variables; however, it is known that one can
approximately solve it to any desired accuracy by designing a polynomial time (approximate)
separation algorithm for the dual~\cite{bansal2006santa}. For our combinatorial algorithm, the dual
shall play an important role in our analysis.  By associating the sets of variables $\{y_i\}_{i\in
  \cP}$ and $\{z_j\}_{j\in \cR}$ to the constraints in the primal corresponding to players and
resources respectively, and letting the primal have the objective function of minimizing the zero
function, we obtain the dual of $CLP(\tau)$ shown in the right box of Figure~\ref{fig:lp}.
\iffalse
% \begin{figure}
\begin{equation*}
\boxed{%
%  \addtolength{\linewidth}{-2\fboxsep}%
%  \addtolength{\linewidth}{-2\fboxrule}%
        \begin{minipage}{10cm}%
          \begin{align*}
            \max & \mbox{  } \sum_{i\in \players} y_i - \sum_{j\in \res} z_j  && \\[2mm]
            y_i &\leq             \sum_{j\in C} z_j, &  \forall i\in \players, \forall C \in
            \conf{i, \tau}\\[1mm]
            y,z & \geq 0.
          \end{align*}
          \vspace{-0.4cm}
        \end{minipage}%
}
\end{equation*}
\fi
%%% Local Variables:
%%% mode: latex
%%% TeX-master: "main"
%%% End: 

% \input{prelim}

\section{Polynomial time algorithm}\label{sec:simplealgo2}
To illustrate our key ideas we first describe a
simpler algorithm that works on
\emph{clustered} instances. This setting, while equivalent to the
general problem up to constant factors, allows for a simpler exposition
of our key ideas. Specifically, we will prove the following theorem in
this section.

\begin{theorem}\label{thm:simplealgo2}
  There is a polynomial time $36$-approximation for \problemmacro{restricted max-min fair allocation}.
\end{theorem}

We note, however, that producing such clustered instances requires
solving the Configuration-LP. To avoid solving it,
and get a purely combinatorial algorithm, we will show how to bypass the
clustering step in Section~\ref{sec:combinatorialalgorithm}.

Before describing our algorithm formally, we begin by giving an informal overview of how it works, while pointing out the key ideas
behind it.

% \footnote{Specifically we run Algorithm~\ref{alg1} on any guess
% $\tau$ with the modification that we terminate when the conclusion of
% Lemma~\ref{lem:manyaddableedges1} fails to hold}
\subsection{Intuitive Algorithm Description and Main Ideas}
\label{sec:simpleint}
Our first step towards recovering an approximate solution to an
instance of \problemmacro{restricted max-min fair allocation}, is
guessing the value of the Configuration-LP $\tau_{OPT}$ by performing
a binary search over the range of its possible values. For a
particular guess $\tau$, assuming that $CLP(\tau)$ is feasible, our
goal now is to {\em approximately} satisfy each player. That is, we will allocate for each player a disjoint
collection of resources, whose value for that player is at least $\tau/36$.
Towards this end, we design a local search procedure, that we
will apply iteratively in order to find such a $36$-approximate
allocation. The input to this procedure will be a partial allocation
that satisfies some (possibly empty) subset of the players, and an
unsatisfied player. Then, our local search procedure will extend the
allocation in order to satisfy the input player as well; hence,
applying this procedure iteratively will satisfy all the players.

We now illustrate some key aspects of this local search procedure
through an example that appears in Figure~\ref{fig:cluster-sketch};
for simplicity, in this example we consider only resources of value
less than $\tau/36$. Given a partial allocation
of resources to a subset of the players, we wish to extend this to
satisfy an additional player $p$.  If there are free resources (i.e., not already
appearing in our partial allocation) of total value $\tau/36$ for $p$,
then we just satisfy $p$ by assigning those resources to her.
Otherwise, we find a set of resources whose value for $p$ is at least
$2\tau/5$; these resources constitute a bundle (or, as we will refer
to it later on, an {\em edge}) $e_p$ we would wish to include in our
partial allocation in order to satisfy $p$. However, we cannot include
this edge right away because there already exist edges in our partial
allocation that share resources with $e_p$; in other words, such
edges are {\em blocking} the inclusion of $e_p$ into our partial
allocation. In Figure~\ref{fig:cluster-sketch}(a), $e_p$ is the gray
 edge, and its blocking edges are the white ones.

At this point, we should make note of the fact that the size of $e_p$ is
considerably larger than our goal of $\tau/36$; this is by design and
due to our {\em greedy} strategy. By considering edges whose size
exceeds our goal, we are able to increase the rate at which blocking
edges are inserted into our local search; indeed, in
Figure~\ref{fig:cluster-sketch}(a), a single greedily-constructed edge
($e_p$) introduced 3 blocking edges. Ultimately, this will
allow us to bound the running time of our local search.

Now, since our goal is to include $e_p$ in our partial allocation, we
are required to free up some of $e_p$'s resources by finding an
alternative way of satisfying the players included in $e_p$'s blocking
edges. The steps we take towards this end appear in
Figure~\ref{fig:cluster-sketch}(b): for each player in $e_p$'s
blocking edges, we find a new edge that we would wish to include into
our partial allocation. But these new gray edges might also be blocked
by existing edges in our partial allocation. Therefore this step
introduces a second {\em layer} of edges comprising a set of edges we
would like to include in our allocation, and their corresponding
blocking edges; these layers are separated by dashed lines in the
example.

Next, we observe that 2 of the 3 gray edges in the second layer
actually have a lot of resources that do not appear in any blocking
edge. In this case, as one can see in
Figure~\ref{fig:cluster-sketch}(c), we select a subset of free resources
from each edge of size at least $\tau/36$ (drawn with dashed lines), and swap these edges for
the existing white edges in our partial allocation. We call
this operation a {\em collapse} of the first layer, only to be left
with $e_p$ and a single blocking edge in the first layer. The way we
decide when to collapse a layer, is dictated by our strategy of {\em lazy updates}:
similar to Figure~\ref{fig:cluster-sketch}(c), we will only collapse a
layer if that would mean that a large fraction of its
blocking edges will be removed.

Finally, in Figure~\ref{fig:cluster-sketch}(d), a significant amount
of resources of $e_p$ has now been freed up. Then, we choose a subset of
these resources (again, drawn with a dashed line), and allocate them
to $p$. At this point, we have satisfied $p$, and managed to extend
our partial allocation to satisfy one more player.

\begin{figure*}[t!]
\begin{minipage}[t]{\linewidth}
    \begin{tikzpicture}[scale=0.3]

\draw [draw=white] (-12,0) -- (-10,0);
%%%% TOP LEFT%%%%%%
\begin{scope}
\draw (-3,0) rectangle (16,18);
\draw [draw=white] (6,0) -- (6,0) node[anchor=north] {(a)};
\draw [rounded corners=1.5mm,fill=lightgray]  (6,6) -- (9,6) -- (6,0) -- (3,6) -- (6,6);
\filldraw (5.85,2.3) rectangle (6.15,2) node[anchor=south] {$p$};
\draw (3.5,5.5)  [fill=black]  circle [radius=0.12];
\draw (4,5.5)  [fill=black]  circle [radius=0.12];
\draw (4.5,5.5)  [fill=black]  circle [radius=0.12];
\draw (5,5.5)  [fill=black]  circle [radius=0.12];
\draw (5.5,5.5)  [fill=black]  circle [radius=0.12];
\draw (6,5.5)  [fill=black]  circle [radius=0.12];
\draw (6.5,5.5)  [fill=black]  circle [radius=0.12];
\draw (7,5.5)  [fill=black]  circle [radius=0.12];
\draw (7.5,5.5)  [fill=black]  circle [radius=0.12];
\draw (8,5.5)  [fill=black]  circle [radius=0.12];
\draw (8.5,5.5)  [fill=black]  circle [radius=0.12];

\begin{scope}[shift={(0.05,0)}]
\draw [rounded corners=1.5mm]  (9.,4.7) -- (10.35,4.7) -- (9.0,10.7) -- (7.65,4.7) -- (9.0,4.7);
\filldraw (8.85,8.7) rectangle (9.15,9);
\end{scope}
\draw (9,5.5)  [fill=black]  circle [radius=0.12];
\draw (9.5,5.5)  [fill=black]  circle [radius=0.12];
\draw (10,5.5)  [fill=black]  circle [radius=0.12];

\begin{scope}[shift={(-3,0)}]
\draw [rounded corners=1.5mm]  (9.,4.7) -- (10.35,4.7) -- (9.0,10.7) -- (7.65,4.7) -- (9.0,4.7);
\filldraw (8.85,8.7) rectangle (9.15,9);
\end{scope}

\begin{scope}[shift={(-6,0)}]
\draw [rounded corners=1.5mm]  (9.,4.7) -- (10.35,4.7) -- (9.0,10.7) -- (7.65,4.7) -- (9.0,4.7);
\filldraw (8.85,8.7) rectangle (9.15,9);
\end{scope}

\draw (3,5.5)  [fill=black]  circle [radius=0.12];
\draw (2.5,5.5)  [fill=black]  circle [radius=0.12];
\draw (2,5.5)  [fill=black]  circle [radius=0.12];
\end{scope}

%%%% TOP RIGHT %%%%%%%%
\begin{scope}[shift={(20,0)}]
\draw (-3,0) rectangle (16,18);

\draw [draw=white] (6,0) -- (6,0) node[anchor=north] {(b)};
\begin{scope}[shift={(0,7)}]
\draw [rounded corners=1.5mm,fill=lightgray]  (6,6) -- (8.5,6) -- (6,0) -- (3.5,6) -- (6,6);
\filldraw (5.85,1.7) rectangle (6.15,2);
\draw (4,5.5)  [fill=black]  circle [radius=0.12];
\draw (4.5,5.5)  [fill=black]  circle [radius=0.12];
\draw (5,5.5)  [fill=black]  circle [radius=0.12];
\draw (5.5,5.5)  [fill=black]  circle [radius=0.12];
\draw (6,5.5)  [fill=black]  circle [radius=0.12];
\draw (6.5,5.5)  [fill=black]  circle [radius=0.12];
\draw (7,5.5)  [fill=black]  circle [radius=0.12];
\draw (7.5,5.5)  [fill=black]  circle [radius=0.12];
\draw (8,5.5)  [fill=black]  circle [radius=0.12];
\end{scope}

\begin{scope}[shift={(3,7)}]
\draw [rounded corners=1.5mm,fill=lightgray]  (8,6) -- (11,6) -- (5.65,0) -- (6,6) -- (10,6);
\filldraw (5.85,1.7) rectangle (6.15,2);
\draw (6.5,5.5)  [fill=black]  circle [radius=0.12];
\draw (7,5.5)  [fill=black]  circle [radius=0.12];
\draw (7.5,5.5)  [fill=black]  circle [radius=0.12];
\draw (8,5.5)  [fill=black]  circle [radius=0.12];
\draw (8.5,5.5)  [fill=black]  circle [radius=0.12];
\draw (9,5.5)  [fill=black]  circle [radius=0.12];
\draw (9.5,5.5)  [fill=black]  circle [radius=0.12];
\draw (10,5.5)  [fill=black]  circle [radius=0.12];
\draw (10.5,5.5)  [fill=black]  circle [radius=0.12];
\draw (11,5.5)  [fill=black]  circle [radius=0.12];
\draw (11.5,5.5)  [fill=black]  circle [radius=0.12];
\draw (12,5.5)  [fill=black]  circle [radius=0.12];
\end{scope}

\begin{scope}[shift={(-3,7)}]
\draw [rounded corners=1.5mm,fill=lightgray]  (3,6) -- (1,6) -- (6.35,0) -- (6,6) -- (3,6);
\filldraw (5.85,1.7) rectangle (6.15,2);
\draw (2,5.5)  [fill=black]  circle [radius=0.12];
\draw (2.5,5.5)  [fill=black]  circle [radius=0.12];
\draw (3,5.5)  [fill=black]  circle [radius=0.12];
\draw (3.5,5.5)  [fill=black]  circle [radius=0.12];
\draw (4,5.5)  [fill=black]  circle [radius=0.12];
\draw (4.5,5.5)  [fill=black]  circle [radius=0.12];
\draw (5,5.5)  [fill=black]  circle [radius=0.12];
\draw (5.5,5.5)  [fill=black]  circle [radius=0.12];
\draw (6.25,5.5)  [fill=black]  circle [radius=0.12];
\end{scope}

\begin{scope}[shift={(-6,7)}]
\draw [rounded corners=1.5mm]  (9.,4.7) -- (10.35,4.7) -- (9.0,10.7) -- (7.65,4.7) -- (9.0,4.7);
\filldraw (8.85,8.7) rectangle (9.15,9);
\end{scope}

\begin{scope}[shift={(-9.9,7)}]
\draw [rounded corners=1.5mm]  (9.,4.7) -- (10.35,4.7) -- (9.0,10.7) -- (7.65,4.7) -- (9.0,4.7);
\filldraw (8.85,8.7) rectangle (9.15,9);
\end{scope}

\begin{scope}[shift={(5,7)}]
\draw [rounded corners=1.5mm]  (9.,4.7) -- (10.35,4.7) -- (9.0,10.7) -- (7.65,4.7) -- (9.0,4.7);
\filldraw (8.85,8.7) rectangle (9.15,9);
\end{scope}

\draw [rounded corners=1.5mm,fill=lightgray]  (6,6) -- (9,6) -- (6,0) -- (3,6) -- (6,6);
\filldraw (5.85,1.7) rectangle (6.15,2)  node[anchor=south] {$p$};
\draw (3.5,5.5)  [fill=black]  circle [radius=0.12];
\draw (4,5.5)  [fill=black]  circle [radius=0.12];
\draw (4.5,5.5)  [fill=black]  circle [radius=0.12];
\draw (5,5.5)  [fill=black]  circle [radius=0.12];
\draw (5.5,5.5)  [fill=black]  circle [radius=0.12];
\draw (6,5.5)  [fill=black]  circle [radius=0.12];
\draw (6.5,5.5)  [fill=black]  circle [radius=0.12];
\draw (7,5.5)  [fill=black]  circle [radius=0.12];
\draw (7.5,5.5)  [fill=black]  circle [radius=0.12];
\draw (8,5.5)  [fill=black]  circle [radius=0.12];
\draw (8.5,5.5)  [fill=black]  circle [radius=0.12];

\begin{scope}[shift={(0.05,0)}]
\draw [rounded corners=1.5mm]  (9.,4.7) -- (10.35,4.7) -- (9.0,10.7) -- (7.65,4.7) -- (9.0,4.7);
\filldraw (8.85,8.7) rectangle (9.15,9);
\end{scope}
\draw (9,5.5)  [fill=black]  circle [radius=0.12];
\draw (9.5,5.5)  [fill=black]  circle [radius=0.12];
\draw (10,5.5)  [fill=black]  circle [radius=0.12];

\begin{scope}[shift={(-3,0)}]
\draw [rounded corners=1.5mm]  (9.,4.7) -- (10.35,4.7) -- (9.0,10.7) -- (7.65,4.7) -- (9.0,4.7);
\filldraw (8.85,8.7) rectangle (9.15,9);
\end{scope}

\begin{scope}[shift={(-6,0)}]
\draw [rounded corners=1.5mm]  (9.,4.7) -- (10.35,4.7) -- (9.0,10.7) -- (7.65,4.7) -- (9.0,4.7);
\filldraw (8.85,8.7) rectangle (9.15,9);
\end{scope}

\draw (3,5.5)  [fill=black]  circle [radius=0.12];
\draw (2.5,5.5)  [fill=black]  circle [radius=0.12];
\draw (2,5.5)  [fill=black]  circle [radius=0.12];
\draw[dashed] (-2.5,1.85) -- (15.5,1.85);
\draw[dashed] (-2.5,8.85) -- (15.5,8.85);
\draw[dashed] (-2.5,15.85) -- (15.5,15.85);

\end{scope}

%%%BOT LEFT%%%%
\begin{scope}[shift={(0,-20)}]
\draw (-3,0) rectangle (16,18);
\draw [draw=white] (6,0) -- (6,0) node[anchor=north] {(c)};
\begin{scope}[shift={(0,7)}]
\draw [rounded corners=1.5mm,fill=lightgray]  (6,6) -- (8.5,6) -- (6,0) -- (3.5,6) -- (6,6);
\begin{scope}[shift={(0,.5)}]
\draw [rounded corners=1.5mm,dashed]  (6,6) -- (7.6,6) -- (6,0) -- (4.5,6) -- (6,6);
\end{scope}
\filldraw (5.85,1.7) rectangle (6.15,2);
\draw (4,5.5)  [fill=black]  circle [radius=0.12];
\draw (4.5,5.5)  [fill=black]  circle [radius=0.12];
\draw (5,5.5)  [fill=black]  circle [radius=0.12];
\draw (5.5,5.5)  [fill=black]  circle [radius=0.12];
\draw (6,5.5)  [fill=black]  circle [radius=0.12];
\draw (6.5,5.5)  [fill=black]  circle [radius=0.12];
\draw (7,5.5)  [fill=black]  circle [radius=0.12];
\draw (7.5,5.5)  [fill=black]  circle [radius=0.12];
\draw (8,5.5)  [fill=black]  circle [radius=0.12];
\end{scope}

\begin{scope}[shift={(3,7)}]
\draw [rounded corners=1.5mm,fill=lightgray]  (8,6) -- (11,6) -- (5.65,0) -- (6,6) -- (10,6);
\begin{scope}[shift={(0,.5)}]
\draw [rounded corners=1.5mm,dashed]  (8,6) -- (9.4,6) -- (5.65,0) -- (6.2,6) -- (8,6);
\end{scope}
\filldraw (5.85,1.7) rectangle (6.15,2);
\draw (6.5,5.5)  [fill=black]  circle [radius=0.12];
\draw (7,5.5)  [fill=black]  circle [radius=0.12];
\draw (7.5,5.5)  [fill=black]  circle [radius=0.12];
\draw (8,5.5)  [fill=black]  circle [radius=0.12];
\draw (8.5,5.5)  [fill=black]  circle [radius=0.12];
\draw (9,5.5)  [fill=black]  circle [radius=0.12];
\draw (9.5,5.5)  [fill=black]  circle [radius=0.12];
\draw (10,5.5)  [fill=black]  circle [radius=0.12];
\draw (10.5,5.5)  [fill=black]  circle [radius=0.12];
\draw (11,5.5)  [fill=black]  circle [radius=0.12];
\draw (11.5,5.5)  [fill=black]  circle [radius=0.12];
\draw (12,5.5)  [fill=black]  circle [radius=0.12];
\end{scope}

\begin{scope}[shift={(-3,7)}]
\draw [rounded corners=1.5mm,fill=lightgray]  (3,6) -- (1,6) -- (6.35,0) -- (6,6) -- (3,6);
\filldraw (5.85,1.7) rectangle (6.15,2);
\draw (2,5.5)  [fill=black]  circle [radius=0.12];
\draw (2.5,5.5)  [fill=black]  circle [radius=0.12];
\draw (3,5.5)  [fill=black]  circle [radius=0.12];
\draw (3.5,5.5)  [fill=black]  circle [radius=0.12];
\draw (4,5.5)  [fill=black]  circle [radius=0.12];
\draw (4.5,5.5)  [fill=black]  circle [radius=0.12];
\draw (5,5.5)  [fill=black]  circle [radius=0.12];
\draw (5.5,5.5)  [fill=black]  circle [radius=0.12];
\draw (6.25,5.5)  [fill=black]  circle [radius=0.12];
\end{scope}

\begin{scope}[shift={(-6,7)}]
\draw [rounded corners=1.5mm]  (9.,4.7) -- (10.35,4.7) -- (9.0,10.7) -- (7.65,4.7) -- (9.0,4.7);
\filldraw (8.85,8.7) rectangle (9.15,9);
\end{scope}

\begin{scope}[shift={(-9.9,7)}]
\draw [rounded corners=1.5mm]  (9.,4.7) -- (10.35,4.7) -- (9.0,10.7) -- (7.65,4.7) -- (9.0,4.7);
\filldraw (8.85,8.7) rectangle (9.15,9);
\end{scope}

\begin{scope}[shift={(5,7)}]
\draw [rounded corners=1.5mm]  (9.,4.7) -- (10.35,4.7) -- (9.0,10.7) -- (7.65,4.7) -- (9.0,4.7);
\filldraw (8.85,8.7) rectangle (9.15,9);
\end{scope}

\draw [rounded corners=1.5mm,fill=lightgray]  (6,6) -- (9,6) -- (6,0) -- (3,6) -- (6,6);
\filldraw (5.85,1.7) rectangle (6.15,2)  node[anchor=south] {$p$};
\draw (3.5,5.5)  [fill=black]  circle [radius=0.12];
\draw (4,5.5)  [fill=black]  circle [radius=0.12];
\draw (4.5,5.5)  [fill=black]  circle [radius=0.12];
\draw (5,5.5)  [fill=black]  circle [radius=0.12];
\draw (5.5,5.5)  [fill=black]  circle [radius=0.12];
\draw (6,5.5)  [fill=black]  circle [radius=0.12];
\draw (6.5,5.5)  [fill=black]  circle [radius=0.12];
\draw (7,5.5)  [fill=black]  circle [radius=0.12];
\draw (7.5,5.5)  [fill=black]  circle [radius=0.12];
\draw (8,5.5)  [fill=black]  circle [radius=0.12];
\draw (8.5,5.5)  [fill=black]  circle [radius=0.12];

\begin{scope}[shift={(0.05,0)}]
\draw [rounded corners=1.5mm]  (9.,4.7) -- (10.35,4.7) -- (9.0,10.7) -- (7.65,4.7) -- (9.0,4.7);
\filldraw (8.85,8.7) rectangle (9.15,9);
\end{scope}
\draw (9,5.5)  [fill=black]  circle [radius=0.12];
\draw (9.5,5.5)  [fill=black]  circle [radius=0.12];
\draw (10,5.5)  [fill=black]  circle [radius=0.12];

\begin{scope}[shift={(-3,0)}]
\draw [rounded corners=1.5mm]  (9.,4.7) -- (10.35,4.7) -- (9.0,10.7) -- (7.65,4.7) -- (9.0,4.7);
\filldraw (8.85,8.7) rectangle (9.15,9);
\end{scope}

\begin{scope}[shift={(-6,0)}]
\draw [rounded corners=1.5mm]  (9.,4.7) -- (10.35,4.7) -- (9.0,10.7) -- (7.65,4.7) -- (9.0,4.7);
\filldraw (8.85,8.7) rectangle (9.15,9);
\end{scope}

\draw (3,5.5)  [fill=black]  circle [radius=0.12];
\draw (2.5,5.5)  [fill=black]  circle [radius=0.12];
\draw (2,5.5)  [fill=black]  circle [radius=0.12];

\end{scope}
%%%BOT RIGHT%%%
\begin{scope}[shift={(20,-20)}]
\draw (-3,0) rectangle (16,18);
\draw [draw=white] (6,0) -- (6,0) node[anchor=north] {(d)};
\draw [rounded corners=1.5mm,fill=lightgray]  (6,6) -- (9,6) -- (6,0) -- (3,6) -- (6,6);
\filldraw (5.85,2.3) rectangle (6.15,2) node[anchor=south] {$p$};
\draw (3.5,5.5)  [fill=black]  circle [radius=0.12];
\draw (4,5.5)  [fill=black]  circle [radius=0.12];
\draw (4.5,5.5)  [fill=black]  circle [radius=0.12];
\draw (5,5.5)  [fill=black]  circle [radius=0.12];
\draw (5.5,5.5)  [fill=black]  circle [radius=0.12];
\draw (6,5.5)  [fill=black]  circle [radius=0.12];
\draw (6.5,5.5)  [fill=black]  circle [radius=0.12];
\draw (7,5.5)  [fill=black]  circle [radius=0.12];
\draw (7.5,5.5)  [fill=black]  circle [radius=0.12];
\draw (8,5.5)  [fill=black]  circle [radius=0.12];
\draw (8.5,5.5)  [fill=black]  circle [radius=0.12];

\begin{scope}[shift={(-6,0)}]
\draw [rounded corners=1.5mm]  (9.,4.7) -- (10.35,4.7) -- (9.0,10.7) -- (7.65,4.7) -- (9.0,4.7);
\filldraw (8.85,8.7) rectangle (9.15,9);
\end{scope}

\draw (3,5.5)  [fill=black]  circle [radius=0.12];
\draw (2.5,5.5)  [fill=black]  circle [radius=0.12];
\draw (2,5.5)  [fill=black]  circle [radius=0.12];

\begin{scope}[shift={(0,1)}]
\draw [rounded corners=1.5mm,dashed]  (6,6) -- (7.7,6) -- (6,0) -- (4.4,6) -- (6,6);
\end{scope}

\end{scope}

       \end{tikzpicture}
    \end{minipage}%
\caption{An example execution of our local search algorithm. In this figure, boxes correspond to players, and circles correspond to resources.}
    
    \label{fig:cluster-sketch}

\end{figure*}
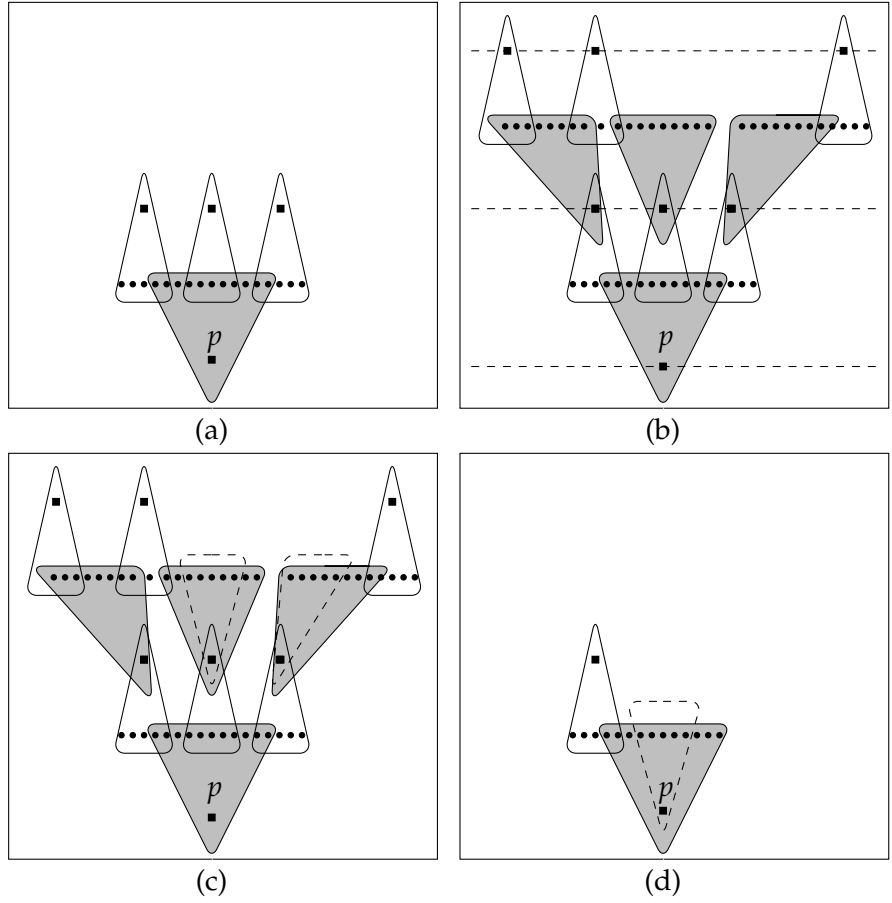

We proceed by formally defining and analyzing the
local search algorithm we sketched above.

% The final combinatorial algorithm in
% Section~\ref{sec:combinatorialalgorithm} contains some more ideas in
% order to overcome issues created by resources of large value.

\subsection{Parameters} Let $\tau > 0$ be a guess on the value of the
Configuration-LP. Our algorithm will use the following setting of parameters:
\begin{align}
  \begin{split}
    \beta &:= 36, \\
    \alpha &:= 5/2, \\
    \mu &:= 1/500.
  \end{split}\label{eqn:parametersforsimplealgo}
\end{align}
Here, $\beta$ is the approximation guarantee, $\alpha$ determines the
``greediness'' of the edges introduced into the layers, and $\mu$
determines the ``laziness'' of the updates of our algorithm. As our
goal is to expose the main ideas, we have not optimized the constants
in this section.

% \begin{itemize}
%   \item $\beta = 36$, which is the obtained approximation guarantee\footnote{As our goal is to expose the main ideas, we have not optimized the constants in this
%   section.}. 
% 	\item $\alpha = 5/2$, which determines the ``greediness'' of our algorithm. 
% 	\item $\mu = 1/500$, which determines the ``laziness'' of the updates of our algorithm.
% \end{itemize}
We shall show that whenever $CLP(\tau)$ is feasible, our algorithm will terminate
with a solution of value at least $\tau/\beta$ for the given instance of
\problemmacro{restricted max-min fair allocation}. Combining this with
a standard binary search then yields a $\beta$-approximation algorithm.

\subsection{Thin and fat edges, and matchings}\label{subsec:thinandfat}
We partition the resource set $\cR$ into
$\cR_f := \{i\in \cR: v_i \geq \tau/\beta\}$ and
$\cR_t := \{i\in \cR \mid v_i < \tau/\beta\}$, fat and thin resources
respectively. Note that in a $\beta$-approximate solution, a player is
satisfied if she is assigned a single fat resource whereas she needs
several thin resources. We will call a pair $(p,R)$, for any
$p\in\players$ and $R\subseteq\cR$ such that $v_p(R) = v(R)$ where $v(R) = \sum_{j\in R} v_j$, an {\em
  edge}. Notice that this definition implies that 
  every resource in $R$ is a resource that player $p$
is interested in. We now define \emph{thin} and \emph{fat edges}.

\begin{definition}[Thin and fat edges]
We will call an edge $(p,R)$, where $p\in\cP$ and $R\subseteq\cR$, {\em fat}, if $\{j\}=R\subseteq \cR_f $ contains a single fat resource
that $p$ is interested in; this already implies that $v_p(R) \geq \tau/\beta$.
On the other hand, we will call an edge $(p,R)$, where $p\in\cP$ and $R\subseteq\cR$, {\em thin}, if $R\subseteq \cR_t$ is a
set of thin resources that $p$ is interested in.

Finally, for any $\delta \geq 1$, we will call an edge $(p,R)$, where $p\in\cP$ and $R\subseteq\cR$, a $\delta$-edge, if
$R$ is a minimal set  (by inclusion) of resources  of value at least $\tau/\delta$ for $p$, i.e.,  $v_p(R) \geq
\tau/\delta$. 
\end{definition}
\begin{remark} 
  A thin $\delta$-edge has value at most $\tau/\delta + \tau/\beta$ due to the
  minimality of the edge. 
%   Also, note that a fat $\delta$-edge may not
%   have resources of value at least $\tau/\delta$, but since we are
%   interested only in a $\beta$ approximate allocation, it is
%   sufficient that such an edge
%   contains a single resource of value at least $\tau/\beta$, which is
%   the case.
\end{remark}

As we have already mentioned, the goal of our local search algorithm is to iteratively extend a {\em partial matching}:
\begin{definition}[Matchings]
  A set $M$ of $\beta$-edges is called a \emph{matching} if each player appears
  in at most one edge and the set of resources used by the edges in $M$ are
  pairwise disjoint.  We say that $M$ \emph{matches} a player $p \in \players$
  if there exists an edge in $M$ that contains $p$. Moreover, it is called
  a \emph{perfect matching} if each player is matched by $M$, and otherwise it
  is called a \emph{partial matching}.
\end{definition}

Using the above terminology, our goal is to find a perfect matching
yielding our desired allocation of value $\tau/\beta$.
Our approach will be to show that as long as the matching $M$ does not match all
players in $\cP$ we can extend it to obtain a matching that matches one more player. This ensures that
starting with an empty partial matching and repeating this procedure
$|\cP|$ times we will obtain an allocation of value at least
$\tau/\beta$. Thus, it suffices to develop such an
algorithm. This is precisely what our  algorithm will
do. 
%As its description requires some concepts to be defined, we first
%formally define these concepts in Section~\ref{subsec:addable}. 
We first state a preprocessing step in Section~\ref{subsec:cluster} before
describing the  algorithm in
Section~\ref{subsec:descriptionofsimplealgo}.% \Authorcomment{Christos}{Update this.}

\subsection{Clustering step}\label{subsec:cluster}
This preprocessing phase produces the
\emph{clustered} instances referred to earlier. The clustering step
that we use is the following reduction due to Bansal and Sviridenko.

\begin{theorem}[Clustering Step~\cite{bansal2006santa}]\label{thm:cluster}
	Assuming that $CLP(\tau)$ is feasible, we can partition the set of players $\cP$ into $m$ clusters $N_1,\dots,N_m$ in polynomial time such that 
	\begin{enumerate}
		\item Each cluster $N_k$ is associated with a distinct subset of $|N_k|-1$ fat items from $\cR_f$
		such that they can be assigned to any subset of  $|N_k|-1$ players in $N_k$, and
		\item there is a feasible solution $x$ to $CLP(\tau)$ such that $\sum_{i \in N_k} \sum_{C \in \mathcal{C}_{t}(i,\tau)} x_{iC}  = 1/2$ for each cluster $N_k$, where $\cC_t(i,\tau)$ denotes the set of configurations for player $i$ comprising only thin items.
	\end{enumerate}
\end{theorem}

Note that the player that is not assigned a fat item can be chosen
arbitrarily and independently for each cluster in the above
theorem. Therefore, after this reduction, it suffices to allocate a
thin $\beta$-edge for one player in each cluster to obtain a
$\beta$-approximate solution for the original instance. Indeed,
Theorem~\ref{thm:cluster} guarantees that we can assign fat edges for
the remaining players. For the rest of the section we assume that our
instance has been grouped into clusters $N_1,\dots,N_m$ by an
application of Theorem~\ref{thm:cluster}. The second property of
these clusters is that each cluster is fractionally
assigned  $1/2$ LP-value of thin configurations. We will use this
to prove the key lemma in this section,
Lemma~\ref{lem:manyaddableedges1}.

We now focus only on allocating one thin $\beta$-edge per cluster and
forget about fat items completely. This makes the algorithm in
Section~\ref{subsec:descriptionofsimplealgo} simpler than our final
combinatorial algorithm, where we also need to handle the
assignment of fat items to players.

\subsection{Description of the algorithm}\label{subsec:descriptionofsimplealgo}

\begin{algorithm}[t!]
%\noindent
%\begin{center}
\begin{small}
\parbox{15.5cm}{

 % {\textnormal \sc{Augmenting Algorithm for Clustered Instances}}\\[-3mm]

{\it Input:}  A partial matching $M$ and an unmatched cluster $N_0$.

{\it Output:} A matching $M'$ that matches all clusters matched by $M$ and also  matches $N_0$.

\begin{enumerate}

  \item \emph{(Initialization)}  Select an arbitrary player $p_0 \in N_0$ and let $A_0 = \emptyset, B_0 = \{(p_0, \emptyset)\}, \ell = 0$, $\cL = (A_0, B_0)$. 
\end{enumerate}
\emph{(Iterative step)} Repeat the following until $N_0$ is matched by $M$:

 \begin{enumerate}\itemsep4mm
  \setcounter{enumi}{1}
\item \emph{(Build phase)} Initialize $A_{\ell+1} =\emptyset$; then for each cluster $N_k$ with a player in $P_\ell$ do:
    \begin{itemize}
      \item  If there is a thin $\alpha$-edge $(p, R)$  with $p\in N_k$ and $R\cap \res(A_{\leq \ell+1} \cup B_{\leq \ell}) = \emptyset$ then \\
        $A_{\ell+1} = A_{\ell+1} \cup \{(p,R)\}.$
        %with $p\in N$ whose
        %resources are disjoint from the resources of the edges in $A_{\leq
        %\ell+1} \cup B_{\leq \ell}$ then add $(p,R)$ to $A_{\ell+1}$.
    \end{itemize}
    At the end of the build phase, let first $B_{\ell+1}$  be the edges of $M$
    that are blocking the edges in $A_{\ell+1}$. Then update the state of the algorithm by appending $(A_{\ell+1}, B_{\ell+1})$ to $\cL$ and by incrementing $\ell$ by
    one.
  \item \emph{(Collapse phase)} While  $\exists t : I_{t+1}  = \{(p, R) \in  A_{t+1} : v_p(R \setminus \res(B_{t+1})) \geq \tau/\beta\}$ has cardinality $\geq \mu|P_t|$: 
    \begin{itemize}
      \item Choose the smallest such $t$.
      
     { \hspace*{-50pt} { \footnotesize\em //We refer to the following steps as collapsing layer $t$.}}
      \item For each cluster $N_k$ with players $q,p$ satisfying $q\in P_t$ and $(p,R)\in I_{t+1}$:
        \begin{itemize}
          \item Let $(q,R_q) \in B_t\cap M$ be the edge containing $q$. 
          \item Replace $(q,R_q)$ in $M$ with  edge $(p,R')$, where $R'$ is a $\tau/\beta$-minimal subset of $R\setminus \res(B_{t+1})$, i.e., update $M\leftarrow M \setminus \{(q,R_q)\} \cup \{(p, R')\}$. 
          \item Finally, remove $(q,R_q)$ from $B_t$.
        \end{itemize}
      \item Discard $(A_i, B_i)$ from $\cL$ for all $i>t$ and set $\ell = t$.
    \end{itemize}
\end{enumerate}
Output the matching $M$ that also matches $N_0$.
}
\end{small}
%\end{center}
\caption{\textsc{Polynomial Time Algorithm for Clustered Instances}}
\label{alg1}
\end{algorithm}

\textbf{Notation:} Recall that it suffices to match exactly one player from each cluster with
a thin $\beta$-edge. With this in mind, we say that a cluster $N_k$ is matched
by $M$ if there exists some player $p \in N_k$ such that $p$ is matched by $M$.
For a set $S$ of edges, we let $\res(S) = \bigcup_{(p,R) \in S} R$ denote the
union of the resources of these edges and we let $\players(S) = \bigcup_{(p, R) \in S} \{p\}$ denote the union of players of these edges. To ease notation, we abbreviate $\players(B_t)$ by $P_t$ in the description of the algorithm and its analysis. Finally, for any family of sets $S_0, S_1,
\ldots, S_\ell$ we  denote $S_0 \cup S_1 \cup \cdots \cup S_t$ by $S_{\leq t}$.

The input to Algorithm~\ref{alg1} is a partial matching $M$ that matches at most one player from each
cluster $N_1,\dots,N_m$, and a cluster $N_0$ that is not matched by $M$; our
goal is to extend our partial matching by matching $N_0$.

The \emph{state of the algorithm} is described by a (dynamic) tuple $(M, \ell, \cL)$,
where  $M$ is the current partial matching and $\cL = ((A_0, B_0), (A_1, B_1), \cdots,
(A_\ell, B_\ell))$ is a list of pairs of sets of ``added'' and ``blocking''
edges that is of length/depth $\ell$. We shall refer to $(A_i, B_i)$ as the
$i$'th \emph{layer}\footnote{The edges of the algorithm naturally form layers
  as described in Section~\ref{sec:simpleint} and as depicted in Figure~\ref{fig:cluster-sketch}.
  The edges in $A_i$ are added so as to try to ``replace'' edges in $B_{i-1}$
in the matching $M$. $B_i$ are then the edges of $M$ that are blocking the
edges in $A_i$.}.

\paragraph{Invariants} The description of our algorithm appears as Algorithm \ref{alg1}. The algorithm is designed to (apart from extending the
matching) maintain the following
invariants at the start of each iterative step: for $i=1, \ldots, \ell$,
  \begin{enumerate}
    \item\label{inv1} $A_i$ is a set of thin $\alpha$-edges that are pairwise
      disjoint, i.e., for two different edges $(p,R) \neq (p', R') \in A_i$  we have $p\neq p'$ and $R\cap R' = \emptyset$.  In addition, each $\alpha$-edge $(p,R) \in A_i$ has $R
      \cap \res(A_{\leq i} \cup B_{\leq i-1} \setminus \{(p,R)\}) = \emptyset$
      (its resources are not shared with edges from earlier iterations or edges
      in $A_i$).

    \item\label{inv2} $B_i =\{(p,R) \in M : \mbox{ $(p,R)$ is blocking an edge  in $A_i$}\}$
      %$B_i = \{(p,R) \in M : \mbox{ there exists $(p',R') \in A_i$  such that $R \cap R' \neq \emptyset$}\}$
      contains those edges of $M$ that blocks edges in $A_i$, where we say that an edge $(p, R) \in M$ blocks an edge $(p', R')$ if $R \cap R' \neq \emptyset$.
    \item\label{inv3} The players of the edges in $A_i$ belong to different clusters and any cluster $N_k$ with a player $p\in N_k$ that  appears in an edge in $A_i$ has a player $q\in N_k$ (that may equal $p$) that
appears in  an edge in $B_{i-1}$.
    \item\label{inv4} $|I_i| < \mu |P_{i-1}|$ where, as in Step 3 of Algorithm~\ref{alg1},
      $I_i  = \{(p, R) \in  A_{i} : v_p(R \setminus \res(B_{i})) \geq
    \tau/\beta\}$ is defined to be those edges that have sufficient amount of
    unblocked resources so as to be added to the matching.
  \end{enumerate}

 In what follows, we further explain the steps of the algorithm and why the
invariants are satisfied. It will then also be clear that the algorithm outputs
an extended matching whenever it terminates. We then analyze its running time
in the next section. 

First the algorithm initializes by selecting an arbitrary player $p_0$ in the cluster $N_0$
that we wish to match.  Then each iteration proceeds in two steps. In the build
phase, the algorithm adds thin $\alpha$-edges (at most one for each cluster with
a player in $P_\ell$) to $A_{\ell+1}$. Notice that the resources of these edges
are disjoint from $\res(A_{\leq \ell} \cup B_{\leq \ell})$ and from each other.
We therefore maintain the first invariant. At the end of the build phase, we define $B_{\ell+1}$ to satisfy the second invariant. The third invariant is also satisfied since we only iterate through the clusters with a player in $P_\ell$ and add at most one edge to $A_{\ell+1}$ for each such cluster.
So after  the build phase, the first three invariants are satisfied. 

The collapse phase will ensure the fourth invariant while not introducing any
violations of the first three.  Indeed, the while-loop runs until the fourth
invariant is satisfied  so we only need to worry about the first three still being satisfied. 
The first and third invariants  remain satisfied because any set $A_i$ that was
affected in the collapse phase is discarded from the algorithm and if $B_i$ was changed then $A_{i+1}$ was also discarded. For the second invariant, note that
after updating the matching $M$, we remove the edge that was removed from the
matching from $B_t$. Hence, $B_t$ still only contains edges of the new
matching that blocks edges in $A_t$. Moreover, by the first invariant, the
newly introduced edge in the matching does not  share any resources with edges
in $A_{\leq t} \cup B_{\leq t}$. Hence, the second invariant also remains true.
Finally, we note that $M$ remains a matching during the update procedure that
matches all clusters that were initially matched. Indeed, when $(q, R_q)$ is removed an edge  $(p,R')$ 
is added to the matching with $p$ being from the same cluster as $q$ (or the algorithm terminates by having successfully matched a player in $N_0$). The added edge is a $\beta$-edge and its resources are disjoint
from all edges in $M$ since (1) $R'$ is a subset of $R\setminus \res(B_{t+1})$, (2)
$B_{t+1}$ contains all blocking edges of $A_{t+1}$ with respect to the matching before the collapse phase, and (3) the edges in $A_{t+1}$ are disjoint so $(p, R')$ is disjoint from any other edges added to the matching in the same collapse phase. We thus
maintain
a valid matching, in which all edges are pairwise disjoint,  and the output is an extended matching that also matches the
cluster $N_0$.

\subsection{Analysis of the algorithm}
We now proceed to show that the algorithm in
Section~\ref{subsec:descriptionofsimplealgo} terminates in polynomial time,
which then implies Theorem~\ref{thm:simplealgo2}. Recall that $\alpha$ is the
parameter that regulates the ``greediness'' of the players  while $\beta$ is
the approximation guarantee, and $\mu$ dictates when we
collapse a layer.

The key lemma that we prove in this section is that in each layer $(A_{i+1}, B_{i+1})$, the number of edges in $A_{i+1}$ is large compared to the number of blocking edges (or, similarly, the number of players) of lower layers. 
%since our algorithm would fail whenever we encounter some layer where $|A_i|=0$, the intuition
%behind the following lemma is that as long as $CLP(\tau)$ is feasible, we will never have to deal with an empty (or, to be more
%precise, too small) layer.
% i.e., at least a constant fraction of
% the players in $P_{i-1}$ are paired with some edge in $A_i$.

\begin{lemma}%[Many addable and immediately edges]
  \label{lem:manyaddableedges1}
  Assuming that $CLP(\tau)$ is feasible, at the beginning of each iterative step, $|A_{i+1}| \geq |P_{\leq i}|/5$ for each $i=0,\dots,\ell-1.$
\end{lemma}

We defer the proof of this statement for now and explain its consequences. As
thin items are of value less than $\tau/\beta = \tau/36$, and each edge in
$A_{\leq \ell}$ is a thin $\alpha$-edge of value at least $\tau/\alpha
= 2\tau/5$, this implies that $B_i$ must be quite large, using $|I_i|
< \mu|P_{i-1}|$ from the fourth invariant. This means that the number of blocking
edges will grow quickly as we prove in the next lemma.

\begin{lemma}[Exponential growth]\label{lem:expgrowth1}
  Assuming that $CLP(\tau)$ is feasible, at the beginning of the iterative step $|P_{i+1}| > 13|P_{\leq i}|/10$ for $i=0,\dots,\ell-1.$
\end{lemma}

\begin{proof}
  Fix an $i$ such that $0 \leq i < \ell$. By the fourth invariant, $|I_{i+1}| < \mu |P_i|$  at the beginning of the iterative step. This means that there are at least
  $|A_{i+1}|-\mu|P_{i}|$ many edges in $A_{i+1}$ which are not in $I_{i+1}$. As each 
  edge in $A_{i+1}\setminus I_{i+1}$   has resources of value at least $\tau/\alpha-\tau/\beta$ that are blocked (i.e., contained in $\res(B_{i+1})$), we can lower bound
  the total value of blocked resources appearing in  $A_{i+1}$  by \[ \left( \frac{\tau}{\alpha} -
    \frac{\tau}{\beta} \right) \left( |A_{i+1}| - \mu |P_{i}| \right).\] Further, since each edge in
  $B_{i+1}$ is of value at most $2\tau/\beta$ by minimality, the total value of such resources is upper
  bounded by $|P_{i+1}|\cdot 2\tau/\beta.$ In total, \[ \left( \frac{\tau}{\alpha} -
    \frac{\tau}{\beta} \right) \left( |A_{i+1}| - \mu |P_{i}| \right) \leq
  |P_{i+1}|\frac{2\tau}{\beta} \implies |P_{ i + 1}| \geq \frac{(\beta - \alpha)(1/5 -
    \mu)}{2\alpha} |P_{\leq i}| > 13|P_{\leq i}|/10, \] where we have used
  Lemma~\ref{lem:manyaddableedges1} to bound $|A_{i+1}|$ by $|P_{\leq i}|/5$ from below.
\end{proof}

Since the number of blocking edges grows exponentially as a function
of the layer index, an
immediate consequence of Lemma~\ref{lem:expgrowth1} is that the total number of
layers in the list $\cL$ at any step in the algorithm is at most $O(\log
|\cP|).$ This means that we have to satisfy the condition in the while-loop of
the collapse phase after at most logarithmically many iterative steps. When
this happens, Algorithm~\ref{alg1} selects the smallest $t$ satisfying the
condition and then proceeds to update $A_{t+1}$ and $B_t$. Note that, by the
condition of the while-loop, and since each edge in $I_{t+1}$ will be updated
in the for-loop (using the third invariant), a constant fraction (at
least $\mu$ as defined in (\ref{eqn:parametersforsimplealgo})) of the edges in $B_t$ are removed. We
refer to these steps of the algorithm as the collapse of layer $t$.
Furthermore, due to the algorithm's first invariant, we know that the edges that compose 
$I_{t+1}$ are pairwise disjoint; therefore, we are able to insert all of them simultaneously into
our matching, which means that the size of our matching does not decrease during the collapse operation. 
On the contrary, if $p_0$ is part of the edges that are inserted into $M$, then we have actually achieved 
to extend our matching $M$.
Intuitively we make large progress whenever we update $M$ during the collapse
of a layer. We prove this by maintaining a signature vector $s :=
(s_0,\dots,s_\ell,\infty)$ during the execution of the algorithm, where $$s_i
:= \lfloor \log_{1/(1-\mu)} |P_i| \rfloor.$$

\begin{lemma}\label{lem:signature}
   The signature vector always reduces in lexicographic value across each iterative step, and the coordinates of the signature vector are always non-decreasing, i.e., $s_0 \leq s_1\dots \leq s_\ell$.
\end{lemma}

\begin{proof}
 Let $s$ and $s'$ be the signature vectors at the beginning and at the
 end of some iterative step.  We now consider two cases depending on
 whether a collapse operation occurs in this iterative step.

  \begin{itemize}
  \item[Case 1.] \textbf{No layer was collapsed.} Clearly, $s' =
    (s_0,\dots,s_\ell,s'_{\ell+1},\infty)$ has smaller lexicographic
    value compared to $s$.

  \item[Case 2.]  \textbf{At least one layer was collapsed.} Let
    $\ell+1$ denote the index corresponding to the newly created layer in
    the build phase. Let $0 \leq t \leq \ell$ be the most recent index
    chosen in the while-loop during the collapse phase. As a result of
    the collapse operation suppose the layer $P_t$ changed to
    $P'_t$. Then we know that $|P'_t| < (1-\mu)|P_t|.$ Since none of
    the layers with indices less than $t$ were affected during this
    procedure, $s' = (s_0,\dots,s_{t-1}, s'_t,\infty)$ where
    $s'_t = \lfloor \log_{1/(1-\mu)} |P'_t| \rfloor \leq \lfloor
    \log_{1/(1-\mu)} |P_t| \rfloor - 1 = s_t - 1.$ This shows that the
    lexicographic value of the signature vector decreases.
  \end{itemize}
  In both cases, the fact that the coordinates of $s'$ are non-decreasing follows from Lemma~\ref{lem:expgrowth1} and the definition of the coordinates of the signature vector.
\end{proof}

Choosing the ``$\infty$'' coordinate of the signature vector to be some value larger than $\log_{1/(1-\mu)} |\cP|$ (so that Lemma~\ref{lem:signature} still holds), we see that each coordinate of the signature vector is at most $U$ and the number of coordinates is also at most $U$ where $U = O(\log |\cP|)$. Thus, the sum of the coordinates of the signature vector is always upper bounded by $U^2$. We now prove that the number of such signature vectors is polynomial in $|\cP|$.

A partition of an integer $N$ is a way of writing $N$ as the sum of
positive integers (ignoring the order of the summands). The number of
partitions of an integer $N$ can be upper bounded by $e^{O(\sqrt{N})}$
by a result of Hardy and Ramanujan~\cite{HardyRam18}\footnote{The asymptotic formula for the number of partitions of $N$ is $ \frac {1} {4N\sqrt3} \exp\left({\pi \sqrt {\frac{2N}{3}}}\right) \mbox { as } N\rightarrow \infty$~\cite{HardyRam18}.}. 
Using that the coordinates of our signature vectors are non-decreasing, each signature vector corresponds to a partition of an integer of value at most $U^2$, and vice versa: given a partition of an integer of size $\ell$, the largest number of the partition will correspond to the $\ell$-th coordinate, the second largest to the $\ell-1$-th coordinate, and so on.
Therefore, we can upper bound the total number of signature vectors by $\sum_{i \leq U^2} e^{O(\sqrt{i})} = |\cP|^{O(1)}.$ Since each iteration of the algorithm takes only polynomial time along with Lemma~\ref{lem:signature} this proves Theorem~\ref{thm:simplealgo2}.

Before we return to the proof of the key lemma in this section,
Lemma~\ref{lem:manyaddableedges1}, let us note an important property
of the algorithm which follows from that in the build-phase we add an
$\alpha$-edge for each cluster as long as it is disjoint from the already
added resources. 

\begin{fact}\label{fact} Let $q$ be a player from some cluster $N_k$. Notice
  that if a player $q$ is part of some blocking edge in the $i^{\text{th}}$ layer, i.e.,
  $q \in P_i$, and further there is no edge $(p,R) \in A_{i+1}$ with $p\in N_k$
  then it means that none of the players in $N_k$ have a set of resources of
  value at least $\tau/\alpha$ disjoint from the resources $\res(B_{\leq i} \cup A_{\leq i+1})$.
\end{fact}

\begin{proof}[Proof of Lemma~\ref{lem:manyaddableedges1}]
  Notice that since the set $A_i$ is discarded if it is modified or
  any of the sets $A_0, A_1, \ldots, A_{i-1}$,
  $B_0, B_1 \ldots, B_{i-1}$ is modified, it is sufficient to verify
  the inequality when we construct the new layer
  $(A_{\ell+1}, B_{\ell+1})$ in the build phase. The proof is now by
  contradiction. Suppose $|A_{\ell+1}| < |P_{\leq \ell}|/5$ after the
  build phase. Let $\cN \subseteq \{N_1,\dots,N_{m}\}$ be the clusters
  that have a player in an edge $B_{\leq \ell}$ but no player in an
  edge in $A_{\leq \ell+1}$. We have that,
  $|\cN| = |P_{\leq \ell}|-|A_{\leq \ell+1}|.$

Recall that $\cC_t(i,\tau)$ denotes the set of configurations for
player $i$ comprising only thin items. By Theorem~\ref{thm:cluster}
there exists an $x$ that is feasible for $CLP(\tau)$ such that
$\sum_{i \in N_k} \sum_{C \in \mathcal{C}_{t}(i,\tau)} x_{iC} = 1/2$
for each cluster $N_k$. Now form the bipartite hypergraph $\cH =
(\cN \cup \cR_t, E)$ where we have vertices for clusters in $\cN$
and thin items in $\cR$, and edges $(N_k, C)$ for every cluster $N_k$
and thin configuration $C$ such that $x_{pC} > 0$ and $p \in
N_k$. To each edge $(N_k, C)$ in $\cH$ assign the weight
$\left(\sum_{i \in N_k} x_{iC}\right)\sum_{j\in C} v_j$. The total
weight of edges in $\cH$ is at least $|\cN|\tau/2.$ 
Let $Z  = \res(B_{\leq \ell} \cup A_{\leq \ell+1})$ denote the thin items appearing in the edges of $\cL$ and $A_{\ell+1}$. Let $v(Z) = \sum_{j\in Z} v_j$
denote their value.
Now remove all these items
from the  
hypergraph to form $\cH'$ which has edges $(N_k, C \setminus Z)$ for
each edge $(N_k, C)$ in $\cH$. The weight of $(N_k, C \setminus Z)$ is
similarly defined to be $\left(\sum_{i \in N_k}
x_{iC}\right)\sum_{j\in C\setminus Z} v_j.$

Let us upper bound the total value of thin items appearing in  $Z$. Consider some layer $(A_j, B_j)$. The total value of resources in thin
$\alpha$-edges in $A_j$ is at most  $(\tau/\alpha + \tau/\beta)|A_j|$ by the
minimality of the edges. The value of resources in $B_j$ not already
present in some edge in $A_j$ is at most $(\tau/\beta)|B_j|$ also by minimality of the thin
$\beta$-edges in $B_j$. Therefore, $v(Z)$ is
 at most
\[ \sum_{j=1}^{\ell}\left( (\frac\tau\alpha + \frac\tau\beta)|A_j| + (\frac\tau\beta)|B_j| \right) + |A_{\ell+1}|\left(\frac\tau\alpha + \frac\tau\beta \right) < |A_{\leq \ell+1}|\left(\frac\tau\alpha + \frac\tau\beta \right) + |P_{\leq \ell}|\frac\tau\beta. \] 

%% \[ \sum_{j=1}^{\ell}\left( (\frac{2\tau}{5} + \frac{\tau}{36})|A_j| + (\frac{\tau}{36})|B_j| \right) + |A_{\ell+1}|\left(\frac{2\tau}{5} + \frac{\tau}{36} \right) < |A_{\leq \ell+1}|\left(\frac{77\tau}{180}\right) + |P_{\leq \ell}|\frac{\tau}{36}. \] As each thin item is fractionally assigned to an extent of at most one by $x$, the total sum of edge weights in $\cH'$ is at least \[ m\tau/2-(|A_{\leq \ell+1}|\left(\frac{2\tau}{5} + \frac{\tau}{36} \right) + |P_{\leq \ell}|\frac{\tau}{36}).\]

As the sum of the edge weights in $\cH$ is at least $(|\cN|/2)(\tau)$, the sum of edge weights in $\cH'$ is at least $|\cN|\tau/2-v(Z)$. And by Fact~\ref{fact}, the sum of edge weights in $\cH'$ must be strictly smaller than $(\cN/2)(\tau/\alpha)$. Thus,

\begin{equation}\label{eqn:prev}
\frac{(|P_{\leq \ell}|- |A_{\leq \ell + 1}|)}{2}\tau - |A_{\leq
  \ell+1}|\left(\frac\tau\alpha + \frac\tau\beta \right) - |P_{\leq
  \ell}|\frac\tau\beta < \frac{(|P_{\leq \ell}|- |A_{\leq \ell +
    1}|)}{2}\frac{\tau}{\alpha}.\tag{*}\end{equation}
%% That is,
%% \[
%% \implies (|P_{\leq \ell}|- |A_{\leq \ell + 1}|)\left(1-\frac1\alpha \right) - 2|A_{\leq
%% \ell+1}|\left(\frac1\alpha + \frac1\beta \right) - 2|P_{\leq \ell}|\frac1\beta \geq 0.\]
Note that $|A_{\leq \ell+1}|$ appears with a larger negative coefficient (in absolute terms) on the left-hand-side than on the
right-hand-side. Therefore, if~\eqref{eqn:prev}  holds then it also holds for an upper bound of
$|A_{\leq \ell+1}|$. We shall compute such a bound and reach a contradiction.

 We start by computing an upper bound on $|A_{j+1}|$ for $j=0,\dots,\ell-1$. The fourth invariant says that except for at most $\mu |P_j|$ edges
 in $A_{j+1}$, the remainder have at least $\tau/\alpha -
 \tau/\beta$ value of resources blocked by the edges in
 $B_{j+1}$. Using this, $$\left( \frac{\tau}{\alpha} -
 \frac{\tau}{\beta} \right) \left( |A_{j+1}| - \mu |P_j| \right) \leq
|P_{j+1}|\frac{2\tau}{\beta} \overset{\text{summing over
     $j$}}{\implies} \left( \frac{\tau}{\alpha} -
 \frac{\tau}{\beta} \right) \left( |A_{\leq \ell}| - \mu |P_{\leq \ell-1}| \right) \leq
 |P_{\leq \ell}|\frac{2\tau}{\beta}.$$

Rearranging terms we have, \[|A_{\leq \ell}| \leq |P_{\leq
   \ell}|\frac{2\alpha}{\beta-\alpha} + \mu |P_{\leq \ell-1}| \leq
 |P_{\leq \ell}|\left( \frac{2\alpha}{\beta-\alpha} + \mu \right).\]

Substituting this upper bound in~(\ref{eqn:prev}) along with our assumption $|A_{\ell+1}| < |P_{\leq \ell}|/5$ we get (after some algebraic manipulations)
\[
  |P_{\leq \ell}|\left(1-\frac1\alpha-\frac2\beta\right) - |P_{\leq
    \ell}|\left( \frac{2\alpha}{\beta-\alpha} + \mu + 1/5
  \right)\left(1+\frac1\alpha+\frac2\beta\right) < 0.\] This is a
contradiction because if we substitute in the values of
$\alpha,\beta,$ and $\mu$ from (\ref{eqn:parametersforsimplealgo})
the left-hand-side is positive.
\end{proof}

\section{Combinatorial Algorithm}
\label{sec:combinatorialalgorithm}

In the previous section, we described a $36$-approximation
algorithm for \problemmacro{restricted max-min fair allocation};
however, this algorithm required us to solve the Configuration-LP. In
this section, we will design and analyze  a purely combinatorial $(6+2\sqrt{10}+\epsilon)
$-approximation algorithm, for any $0<\epsilon\leq 1$ (for reference, note that $6+2\sqrt{10}<13$). This will prove our main result, Theorem \ref{thm:main}.

We start by providing an informal overview of how the combinatorial algorithm works. 

\subsection{Intuitive Algorithm Description}\label{subsec:intuitive}
To begin with, the general framework of our combinatorial algorithm is similar to that of the simpler algorithm we described in Section \ref{sec:simplealgo2}: we guess an optimal value $\tau$ for the Configuration-LP, and we then try to find an allocation of resources which approximately satisfies every player, i.e., assigns to each player a set of resources of total value at least $\tau/13$ for that player.
To do so, we will again design a local search procedure, whose goal will be to extend a given partial allocation of resources, so as to satisfy one more player.

An example execution of our combinatorial algorithm appears in Figure \ref{fig:DPN}: there, given a partial allocation of resources to players, we want to extend this allocation to satisfy player $p$.
Naturally, if there is a set of resources, that do not appear in the given partial allocation, and whose total value for $p$ is at least $\tau/13$, we will assign these resources to player $p$. Otherwise, we find an edge $e_p$ whose total value for $p$ is at least $\tau/2$ (the bottom gray edge in  Figure \ref{fig:DPN}(a)), and consider all the edges in our given partial allocation that share resources with that set (the white edges intersecting $e_p$ in Figure \ref{fig:DPN}(a)); these edges constitute the first layer that is shown in Figure \ref{fig:DPN}(a). 

At this point, we should make note of the fact that, similar to the
simpler algorithm we described in Section \ref{sec:simplealgo2}, we
will again be using a {\em greedy strategy} with respect to the edges
we wish to include in our partial matching. Specifically, even though
we wish to only assign resources of total value at least $\tau/13$ to
each player, the gray edges we attempt to include in our matching are
significantly more valuable (i.e., of total value at least
$\tau/2$). Again, this will imply that every gray edge will intersect
with multiple white/blocking edges, which will eventually help us
prove that the algorithm's running time is polynomial in the size of
the input.

Next, similar to the simpler algorithm we described in Section
\ref{sec:simplealgo2}, we want to free up the resources that appear in
edge $e_p$. We do this by finding disjoint sets of resources that
satisfy the players appearing in the white edges of the first
layer. However, here we encounter the first major difference compared to our previous algorithm: some of the players that appear in the white edges of the first layer can be satisfied by using {\em fat} resources, i.e., resources whose value for their corresponding players is at least $\tau/13$. Since every fat edge we would like to include in our partial allocation can only be blocked by exactly one edge that already belongs to our allocation, {\em alternating paths} of fat edges are created. Such a path, that ends in a gray thin edge, is displayed in Figure \ref{fig:DPN}(b); if we wish to include the gray edge that contains $q_2$ into our partial allocation, then we would have to 
replace the white fat edges with the gray ones.

However, considering such alternating paths of fat edges brings up one issue: since, as is shown in Figure \ref{fig:DPN}(a), the alternating paths that originate at players $p_1$ and $p_2$ end at two distinct gray thin edges, if we were to include both of these edges into our matching, then we would have to guarantee that we will not use the same fat resource to satisfy two different players. In order to do this, we will  include the gray edges that contain players $q_1$ and $q_2$ into our partial allocation, only if the alternating paths that end in these players are \emph{vertex-disjoint}, as is the case in Figure \ref{fig:DPN}(c).

Next, since we have solved the problem of deciding {\em if} we can update our partial matching by replacing white edges with gray ones, the question that arises is {\em when} should we do that. Similar to our simpler algorithm, we will employ the strategy of {\em lazy updates}. In other words, we will be replacing the white edges of some layer with gray ones (or, as we will call this operation, {\em collapse} a specific layer), only if that would mean that a significant amount of the white edges gets replaced. Replacing a significant amount of white (i.e., blocking) edges then implies that we make significant progress towards matching player $p$.

Finally, after we update our partial allocation, by inserting the gray edges containing players $q_1$ and $q_2$, inserting the gray fat edges that belong to the corresponding alternating paths, and removing the white fat edges that belong to  the corresponding alternating paths, we have managed to free up a significant amount of resources of edge $e_p$. Hence, we choose a subset of the resources contained in $e_p$, whose total value is at least $\tau/13$, and include it into our partial allocation. At this point, we have managed to extend our partial allocation to include one more player, namely, player~$p$.

\begin{figure*}[t!]
\begin{minipage}[t]{\linewidth}
    \begin{tikzpicture}[scale=0.25]

 \draw [draw=white] (-15,0) -- (-10,0);
%%%% TOP LEFT%%%%%%
\begin{scope}
\draw (-3,0) rectangle (16,21);

\draw [draw=white] (6,21) -- (6,21) node[anchor=south] {(a)};
\draw [draw=white] (31,21) -- (31,21) node[anchor=south] {(b)};
\draw [draw=white] (6,-23) -- (6,-23) node[anchor=north] {(d)};
\draw [draw=white] (31,-23) -- (31,-23) node[anchor=north] {(c)};
\draw [draw=white] (31.5,0) -- (31.5,0) node[anchor=north] {$\Downarrow$};
\draw [draw=white] (19,10.5) -- (19,10.5) node[anchor=north] {$\Rightarrow$};
\draw [draw=white] (19,-12.5) -- (19,-12.5) node[anchor=north] {$\Leftarrow$};
\begin{scope}[shift={(0,10)}]
\draw [rounded corners=1.5mm,fill=lightgray]  (6,6) -- (7.5,6) -- (6,0) -- (4.5,6) -- (6,6);
\filldraw (5.85,1.7) rectangle (6.15,2);
% \draw (4,5.5)  [fill=black]  circle [radius=0.12];
% \draw (4.5,5.5)  [fill=black]  circle [radius=0.12];
\draw (5,5.5)  [fill=black]  circle [radius=0.12];
\draw (5.5,5.5)  [fill=black]  circle [radius=0.12];
 \draw (6,5.5)  [fill=black]  circle [radius=0.12];
\draw (6.5,5.5)  [fill=black]  circle [radius=0.12];
\draw (7,5.5)  [fill=black]  circle [radius=0.12];
% \draw (7.5,5.5)  [fill=black]  circle [radius=0.12];
% \draw (8,5.5)  [fill=black]  circle [radius=0.12];
\end{scope}

\begin{scope}[shift={(3,10)}]
\draw [rounded corners=1.5mm,fill=lightgray]  (7,6) -- (7.5,6) -- (6,0) -- (4.5,6) -- (7,6);
\filldraw (5.85,1.7) rectangle (6.15,2);
\draw (5.5,5.5)  [fill=black]  circle [radius=0.12];
\draw (6,5.5)  [fill=black]  circle [radius=0.12];
\draw (6.5,5.5)  [fill=black]  circle [radius=0.12];
\draw (7,5.5)  [fill=black]  circle [radius=0.12];
\draw (5,5.5)  [fill=black]  circle [radius=0.12];
\end{scope}
 
\begin{scope}[shift={(-3,7)}]
\draw [rounded corners=1.5mm,fill=lightgray]  (4,9) -- (2,9) -- (6.35,0) -- (7,9) -- (4,9);
\filldraw (5.85,1.7) rectangle (6.15,2);
% \draw (2,8.5)  [fill=black]  circle [radius=0.12];
\draw (2.5,8.5)  [fill=black]  circle [radius=0.12];
\draw (3,8.5)  [fill=black]  circle [radius=0.12];
\draw (3.5,8.5)  [fill=black]  circle [radius=0.12];
 \draw (4,8.5)  [fill=black]  circle [radius=0.12];
\draw (4.5,8.5)  [fill=black]  circle [radius=0.12];
\draw (5,8.5)  [fill=black]  circle [radius=0.12];
 \draw (5.5,8.5)  [fill=black]  circle [radius=0.12];
\draw (6,8.5)  [fill=black]  circle [radius=0.12];
\draw (6.5,8.5)  [fill=black]  circle [radius=0.12];
\end{scope}

% \begin{scope}[shift={(-6,7)}]
% \draw [rounded corners=1.5mm]  (9.,4.7) -- (10.35,4.7) -- (9.0,10.7) -- (7.65,4.7) -- (9.0,4.7);
% \filldraw (8.85,8.7) rectangle (9.15,9);
% \end{scope}
% 
% \begin{scope}[shift={(-9.9,7)}]
% \draw [rounded corners=1.5mm]  (9.,4.7) -- (10.35,4.7) -- (9.0,10.7) -- (7.65,4.7) -- (9.0,4.7);
% \filldraw (8.85,8.7) rectangle (9.15,9);
% \end{scope}
% 
% \begin{scope}[shift={(5,7)}]
% \draw [rounded corners=1.5mm]  (9.,4.7) -- (10.35,4.7) -- (9.0,10.7) -- (7.65,4.7) -- (9.0,4.7);
% \filldraw (8.85,8.7) rectangle (9.15,9);
% \end{scope}

\draw [rounded corners=1.5mm,fill=lightgray]  (6,6) -- (9,6) -- (6,0) -- (3,6) -- (6,6);
\filldraw (5.85,1.7) rectangle (6.15,2)  node[anchor=south] {$p$};
\draw (3.5,5.5)  [fill=black]  circle [radius=0.12];
\draw (4,5.5)  [fill=black]  circle [radius=0.12];
\draw (4.5,5.5)  [fill=black]  circle [radius=0.12];
\draw (5,5.5)  [fill=black]  circle [radius=0.12];
\draw (5.5,5.5)  [fill=black]  circle [radius=0.12];
\draw (6,5.5)  [fill=black]  circle [radius=0.12];
\draw (6.5,5.5)  [fill=black]  circle [radius=0.12];
\draw (7,5.5)  [fill=black]  circle [radius=0.12];
\draw (7.5,5.5)  [fill=black]  circle [radius=0.12];
\draw (8,5.5)  [fill=black]  circle [radius=0.12];
\draw (8.5,5.5)  [fill=black]  circle [radius=0.12];

\begin{scope}[shift={(0.05,0)}]
\draw [rounded corners=1.5mm]  (9.,4.7) -- (10.35,4.7) -- (9.0,10.7) -- (7.65,4.7) -- (9.0,4.7);
\filldraw (8.85,8.7) rectangle (9.15,9) node[anchor=north,yshift=-4] {$p_2$};
\filldraw (8.85,11.7) rectangle (9.15,12) node[anchor=south,yshift=4] {$q_2$};%+3
\draw[dashed] (9,10.35) circle[x radius=1, y radius=2];
\end{scope}
\draw (9,5.5)  [fill=black]  circle [radius=0.12];
\draw (9.5,5.5)  [fill=black]  circle [radius=0.12];
\draw (10,5.5)  [fill=black]  circle [radius=0.12];

\begin{scope}[shift={(-3,0)}]
\draw [rounded corners=1.5mm]  (9.,4.7) -- (10.35,4.7) -- (9.0,10.7) -- (7.65,4.7) -- (9.0,4.7);
\filldraw (8.85,8.7) rectangle (9.15,9) node[anchor=north,yshift=-4] {$p_1$};
\filldraw (8.85,11.7) rectangle (9.15,12) node[anchor=south,yshift=4] {$q_1$};%+3
\draw[dashed] (9,10.35) circle[x radius=1, y radius=2];
\end{scope}

\begin{scope}[shift={(-6,0)}]
\draw [rounded corners=1.5mm]  (9.,4.7) -- (10.35,4.7) -- (9.0,10.7) -- (7.65,4.7) -- (9.0,4.7);
\filldraw (8.85,8.7) rectangle (9.15,9);
\end{scope}

\begin{scope}[shift={(-6,10)}]
\draw [rounded corners=1.5mm]  (9.,4.7) -- (10.35,4.7) -- (9.0,10.7) -- (7.65,4.7) -- (9.0,4.7);
\filldraw (8.85,8.7) rectangle (9.15,9);
\end{scope}

\begin{scope}[shift={(-9,10)}]
\draw [rounded corners=1.5mm]  (9.,4.7) -- (10.35,4.7) -- (9.0,10.7) -- (7.65,4.7) -- (9.0,4.7);
\filldraw (8.85,8.7) rectangle (9.15,9);
\end{scope}

\draw (3,5.5)  [fill=black]  circle [radius=0.12];
\draw (2.5,5.5)  [fill=black]  circle [radius=0.12];
\draw (2,5.5)  [fill=black]  circle [radius=0.12];
\draw[dashed] (-2.5,1.85) -- (15.5,1.85);
\draw[dashed] (-2.5,8.85) -- (15.5,8.85);
\draw[dashed] (-2.5,18.85) -- (15.5,18.85);

\end{scope}

\draw[dashed] (9,8.35) -- (31.5,0.3);
\draw[dashed] (9,12.35) -- (31.5,20.7);
%%%TOP RIGHT
\begin{scope}[shift={(25,0)}]

     \begin{scope}[shift={(0,11)}]
\filldraw[fill=lightgray, draw=black, rounded corners=1.5mm] (4.73,10) -- (6.5,3.5) -- (8.27,10);
 \end{scope}
\draw (-3,0) rectangle (16,21);
 \draw[dashed] (6.5,10.5) circle [x radius=5, y radius=10];

 \begin{scope}[shift={(0,1)}]

  \begin{scope}[shift={(0,1)}]
 \draw[fill=lightgray] (6.5,3.5) circle[x radius=1, y radius=2.5];
 \filldraw (6.35,1.85) rectangle (6.65,2.15) node[anchor=north,yshift=-8] {$p_2$};
 \draw (6.5,5)  [fill=black]  circle [radius=0.12];
% \filldraw (6.35,4.85) rectangle (6.65,5.15);
\draw [rounded corners=1.5mm] (5,-2) -- (6.5,3.5) -- (8,-2);
 \end{scope}
 
 \begin{scope}[shift={(0,7)}]
 \draw[fill=lightgray] (6.5,3.5) circle[x radius=1, y radius=2.5];
 \filldraw (6.35,1.85) rectangle (6.65,2.15);
 \draw (6.5,5)  [fill=black]  circle [radius=0.12];
% \filldraw (6.35,4.85) rectangle (6.65,5.15);
 \end{scope}
 
    \begin{scope}[shift={(0,4)}]
 \draw (6.5,3.5) circle[x radius=1, y radius=2.5];
%  \filldraw (6.35,1.85) rectangle (6.65,2.15);
\filldraw (6.35,4.85) rectangle (6.65,5.15);
 \end{scope}
 
     \begin{scope}[shift={(0,10)}]
 \draw (6.5,3.5) circle[x radius=1, y radius=2.5];
%  \filldraw (6.35,1.85) rectangle (6.65,2.15);
\filldraw (6.35,4.85) rectangle (6.65,5.15) node[anchor=south,yshift=8] {$q_2$};
% \filldraw[fill=lightgray, draw=black, rounded corners=1.5mm] (3,10) -- (6.5,3.5) -- (10,10);
 \end{scope}
 
 \end{scope}

\end{scope}

%%%BOT LEFT
\begin{scope}[shift={(0,-23)}]
\draw (-3,0) rectangle (16,21);

\draw [rounded corners=1.5mm,fill=lightgray]  (6,6) -- (9,6) -- (6,0) -- (3,6) -- (6,6);
\filldraw (5.85,1.7) rectangle (6.15,2)  node[anchor=south] {$p$};
\draw (3.5,5.5)  [fill=black]  circle [radius=0.12];
\draw (4,5.5)  [fill=black]  circle [radius=0.12];
\draw (4.5,5.5)  [fill=black]  circle [radius=0.12];
\draw (5,5.5)  [fill=black]  circle [radius=0.12];
\draw (5.5,5.5)  [fill=black]  circle [radius=0.12];
\draw (6,5.5)  [fill=black]  circle [radius=0.12];
\draw (6.5,5.5)  [fill=black]  circle [radius=0.12];
\draw (7,5.5)  [fill=black]  circle [radius=0.12];
\draw (7.5,5.5)  [fill=black]  circle [radius=0.12];
\draw (8,5.5)  [fill=black]  circle [radius=0.12];
\draw (8.5,5.5)  [fill=black]  circle [radius=0.12];

\begin{scope}[shift={(-6,0)}]
\draw [rounded corners=1.5mm]  (9.,4.7) -- (10.35,4.7) -- (9.0,10.7) -- (7.65,4.7) -- (9.0,4.7);
\filldraw (8.85,8.7) rectangle (9.15,9);
\end{scope}

\draw (3,5.5)  [fill=black]  circle [radius=0.12];
\draw (2.5,5.5)  [fill=black]  circle [radius=0.12];
\draw (2,5.5)  [fill=black]  circle [radius=0.12];
\draw[dashed] (-2.5,1.85) -- (15.5,1.85);
\draw[dashed] (-2.5,8.85) -- (15.5,8.85);
% \draw[dashed] (-2.5,18.85) -- (15.5,18.85);
\begin{scope}[shift={(0,.5)}]
\draw [rounded corners=1.5mm,dashed]  (6,6) -- (7.6,6) -- (6,0) -- (4.5,6) -- (6,6);
\end{scope}
\end{scope}

%%%BOT RIGHT
\begin{scope}[shift={(25,-23)}]
\draw (-3,0) rectangle (16,21);

\begin{scope}[shift={(0,1)}]

\begin{scope}[shift={(0,6)}]
  \draw[fill=lightgray,rotate around={25:(6.5,-1)}] (6.5,-1) circle[x radius=4.25, y radius=1];
\draw (6.5,0.5) circle[x radius=3.75, y radius=1];
 \draw[fill=black] (3.18,0.35) rectangle (3.48,0.65) node[anchor=east,xshift=-4] {$q_1$};
 \draw[fill=black] (9.66,0.5) circle [radius=0.12];
\end{scope}

\begin{scope}[shift={(0,3)}]
  \draw[fill=lightgray,rotate around={25:(6.5,-1)}] (6.5,-1) circle[x radius=4.25, y radius=1];
\draw (6.5,0.5) circle[x radius=3.75, y radius=1];
 \draw[fill=black] (3.18,0.35) rectangle (3.48,0.65);
 \draw[fill=black] (9.66,0.5) circle [radius=0.12];
\end{scope}

\begin{scope}
% \draw[fill=lightgray] (6.5,0.5) circle[x radius=3.75, y radius=1];
 \draw[fill=black] (3.18,0.35) rectangle (3.48,0.65) node[anchor=east,xshift=-4] {$p_1$};
 \draw[fill=black] (9.66,0.5) circle [radius=0.12];
\end{scope}

\begin{scope}[shift={(0,9)}]
 \draw[fill=black] (3.18,0.35) rectangle (3.48,0.65);
 \draw[fill=black] (9.66,0.5) circle [radius=0.12];
\end{scope}

\begin{scope}[shift={(0,12)}]
 \draw[fill=black] (3.18,0.35) rectangle (3.48,0.65) node[anchor=east,xshift=-4] {$q_2$};
 \draw[fill=black] (9.66,0.5) circle [radius=0.12];
\end{scope}

\begin{scope}[shift={(0,15)}]
\draw[fill=lightgray] (6.5,0.5) circle[x radius=3.75, y radius=1];
 \draw[fill=black] (3.18,0.35) rectangle (3.48,0.65);
 \draw[fill=black] (9.66,0.5) circle [radius=0.12];
  \draw[rotate around={25:(6.5,-1)}] (6.5,-1) circle[x radius=4.25, y radius=1];
\end{scope}

\begin{scope}[shift={(0,18)}]
\draw[fill=lightgray] (6.5,0.5) circle[x radius=3.75, y radius=1];
 \draw[fill=black] (3.18,0.35) rectangle (3.48,0.65) node[anchor=east,xshift=-4] {$p_2$};
 \draw[fill=black] (9.66,0.5) circle [radius=0.12];
 \draw[rotate around={25:(6.5,-1)}] (6.5,-1) circle[x radius=4.25, y radius=1];
\end{scope}
\end{scope}
\end{scope}

       \end{tikzpicture}
    \end{minipage}%
\caption{An illustration of our combinatorial algorithm. In this figure, boxes correspond to players and circles correspond to resources.}
    
    \label{fig:DPN}

\end{figure*}
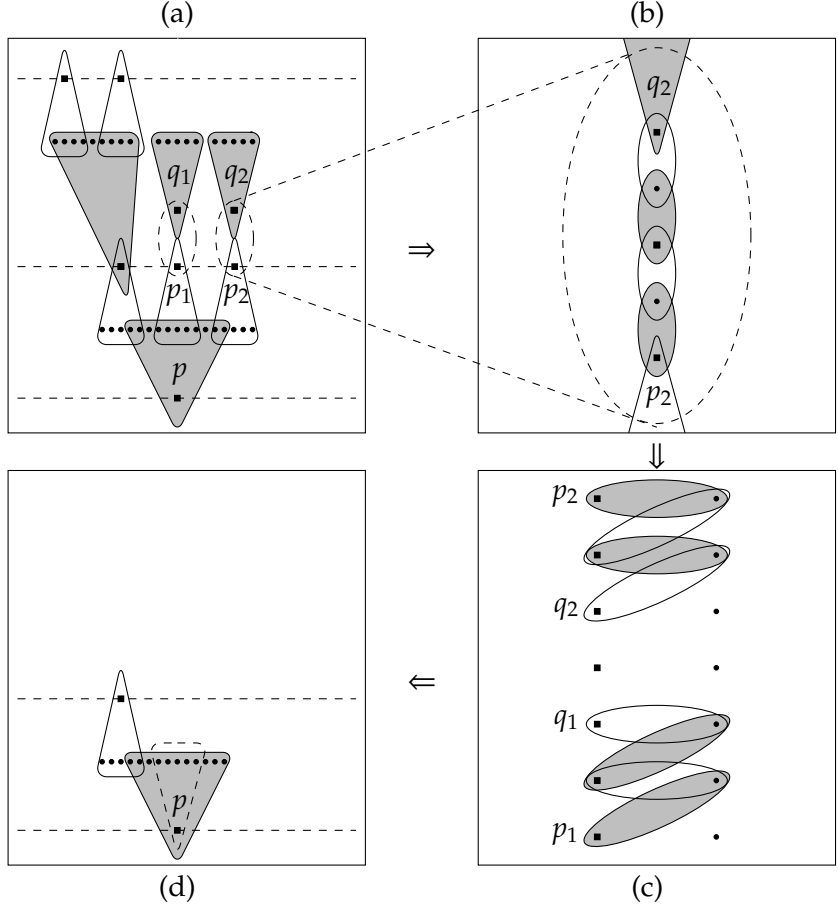

%In the sequel, we proceed with formally defining and analyzing the local search algorithm we sketched above.

\subsection{Parameters}\label{subsec:parameters}
 Let $\tau > 0$ be a guess on the value of the
Configuration-LP, and fix some $0<\epsilon\leq 1$. Our algorithm will use the following setting of parameters:
\begin{align}
\begin{split}
\beta &:= 2(3+\sqrt{10})+\epsilon, \\
\alpha &:= 2, \\
\mu &:= \epsilon/100.
\end{split}
\label{eq:param}
\end{align}
Similar to our simpler algorithm, $\beta$ is the approximation guarantee, $\alpha$ determines the
``greediness'' of the algorithm, and $\mu$
determines the ``laziness'' of the updates of our algorithm.

We shall show that whenever $CLP(\tau)$ is feasible, our algorithm will terminate
with a solution of value at least $\tau/\beta$ for the given instance of
\problemmacro{restricted max-min fair allocation}. Combining this with
a standard binary search then yields a $\beta$-approximation algorithm.

\subsection{Description of the Algorithm}\label{subsec:combalgo-description}
We begin by noting that we will be re-using the definitions of {\em fat} and {\em thin} edges, $\delta$-{\em edges}, and {\em (partial) matchings} that we introduced in Section \ref{subsec:thinandfat}. However, we remind the reader that the parameters we used in the above definitions have now changed, see~\eqref{eq:param}.

The goal of our algorithm will be to find a perfect matching. Similar to our simpler algorithm, the way we do this is by designing an augmenting algorithm, that will extend any given partial matching to satisfy one more player. Thus, starting from an empty matching and iteratively applying the augmenting algorithm will yield a perfect matching that corresponds to a $\beta$-approximate allocation. We remark that for the purposes of our algorithm, any partial matching we consider contains the maximum number of fat resources possible. In order to enforce this condition, we find a maximum matching between fat resources and players; this will be our initial partial matching. Starting from this partial matching, we proceed to iteratively extend it, by matching one more player at a time while never decreasing the number of fat items in our allocation.

We proceed to define the concepts of {\em Disjoint Path Networks} and {\em Canonical Decompositions}, that are necessary to state our combinatorial  algorithm.
These concepts will be used to implement the idea of updating our partial matching using \emph{vertex-disjoint} alternating paths, that we mentioned in Section \ref{subsec:intuitive}. We then  state our algorithm formally, and we analyze its running time in the subsequent sections.

\paragraph{Disjoint Path Networks}
 As we discussed in the overview of our combinatorial algorithm, we need a way to ensure that the alternating paths we use to update our partial matching are disjoint. We say that two paths are disjoint if they are vertex-disjoint. To do so, we employ a structure called {\em Disjoint Path Networks}.

Given a partial matching $M$, let $H_M = (\cP \cup \cR_f, E_M)$ be the directed graph defined as follows:
there is a vertex for each player in $\cP$ and each fat resource in
$\cR_f$; and, there is an arc from a
player in $p \in \cP$ to a fat resource $f \in \cR_f$ if $p$ is interested in $f$ unless the arc
$(p, \{f\})$ appears in $M$ in which case there is an arc $(\{f\}, p)$. Note that
the graph $H_M$ depends only on the assignment of fat resources to players in $M$.

Now, let $S,
T \subseteq \cP$ be a set of sources and sinks respectively that are not necessarily disjoint. Let
$\FLOW{M}{S}{T}$ denote the flow network we get if we place unit capacities on the vertices of $H_M$, and use $S$ and $T$ as sources and sinks respectively. Furthermore, let $\DP{M}{S}{T}$ denote the value of an optimal
solution, i.e., the maximum number of disjoint paths from the sources $S$ to the sinks $T$ in the graph $H_M$.

In our algorithm, $S$ and $T$ will contain only vertices in $H_M$ corresponding to players in
$\cP$. However, to specify a sink we sometimes abuse notation and specify an edge since the corresponding sink vertex
can be deduced from it. For example, if we write $\DP{M}{X}{Y}$, for some set of players $X$ and some set of edges $Y$, then we mean the maximum number of disjoint paths that start at a player in $X$ and end in a player that appears in some edge in $Y$. 

For basic concepts related to flows, such as flow networks and augmenting paths, we refer the reader to the textbook by Cormen,
Leiserson, Rivest and Stein~\cite{DBLP:books/daglib/0023376}.
\paragraph{State of the Algorithm}
The \emph{state of the algorithm} is described by a dynamic tuple $(M, \ell, \cL, I)$,
where  $M$ is the current partial matching, $\cL = ((A_0, B_0,d_0), (A_1, B_1,d_1), \cdots,
(A_\ell, B_\ell,d_\ell))$ is a list of $\ell$ {\em layers} and $I$ is a set of "immediately addable" edges. Each layer $L_i = (A_i,B_i,d_i)$  consists of a set of "added" edges $A_i$, a set of "blocking" edges $B_i$, and a positive integer $d_i$. We note that $d_i$ is redundant for the formal statement of our algorithm, but will be handy in our analysis.

\paragraph{Canonical Decompositions}
We proceed to define the last concept necessary to describe our combinatorial algorithm. Recall that we denote $\cup_{i\leq t} S_i$ by $S_{\leq t}$, for some sequence of sets $S_0,\ldots S_t$, and that $P_i$ denotes the players that appear in $B_i$.  Moreover, for a set $S$ of edges we use $\players(S)$ to denote the set of players that appear in an edge in $S$ and we use $\cR(S)$ to denote the set of resources that appear in an edge in $S$.

\begin{definition}[Canonical Decomposition of $I$]
  \label{def:cancomp}
	Given a state $(M,\ell,\cL, I)$ of the algorithm, we call a collection of disjoint subsets $\{I_0, I_1, \ldots, I_\ell\}$ of $I$ a \emph{canonical
		decomposition} if it satisfies the following conditions:
	\begin{enumerate}
		\item For $i=0, 1, \ldots,\ell$, $|I_{\leq i}| = \DP{M}{P_{\leq i}}{I_{\leq i}} = \DP{M}{P_{\leq i}}{I}$. 
		\item There exists an optimal solution $W$ to $\FLOW{M}{P_{\leq \ell}}{I}$ such that,
		for $i=0,1\ldots, \ell$, $|I_i|$ paths in $W$ go from
                players $Q_i \subseteq P_i$ to the sinks in $I_i$. We
                denote these paths by $W_i$. We also refer to $W$ as
                the {\em canonical solution} corresponding to the
                decomposition. 
	\end{enumerate}
	\end{definition}
	
	As we will see in Section  \ref{subsec:single-iter}, canonical decompositions and their corresponding canonical solutions can be computed in polynomial time.
	
\paragraph{Algorithm Statement}

\begin{algorithm}[t!]
%\noindent
%\begin{center}
\begin{small}
\parbox{15.5cm}{

 % {\textnormal \sc{Augmenting Algorithm for Clustered Instances}}\\[-3mm]

{\it Input:}  A partial matching $M$ and an unmatched player $p_0$.

{\it Output:} A matching $M'$ that matches all players matched by $M$ and also  matches $p_0$.

\begin{enumerate}

  \item \emph{(Initialization)}  Set $A_0 = \emptyset, B_0 = \{(p_0, \emptyset)\}, \ell = 0$, $d_0=0$ and $\cL = (A_0, B_0,d_0)$. 
\end{enumerate}
\emph{(Iterative step)} Repeat the following until $p_0$ is matched by $M$:

 \begin{enumerate}\itemsep4mm
  \setcounter{enumi}{1}
\item \emph{(Build phase)} Initialize $A_{\ell+1} =\emptyset$. 
While there exists a thin $\alpha$-edge $(p,R)$ such that 
$R\cap \res(A_{\leq \ell+1} \cup B_{\leq \ell}\cup I) = \emptyset$
and
 $\DP{M}{P_{\leq \ell}}{A_{\leq \ell+1}\cup I \cup\{(p,R)\}} > \DP{M}{P_{\leq \ell}}{A_{\leq \ell+1}\cup I}$:
    \begin{itemize}
      \item  If $v_p(R\setminus \cR(M))< \tau/\beta$, then set $A_{\ell+1}=A_{\ell+1}\cup\{(p,R)\}$, else set $I=I\cup\{(p,R)\}$.
        %with $p\in N$ whose
        %resources are disjoint from the resources of the edges in $A_{\leq
        %\ell+1} \cup B_{\leq \ell}$ then add $(p,R)$ to $A_{\ell+1}$.
    \end{itemize}
    At the end of the build phase, let  $B_{\ell+1}$  be the edges of $M$
    that are blocking the edges in $A_{\ell+1}$. Set $d_{\ell+1} \leftarrow \DP{M}{P_{\leq l}}{A_{\leq l+1}\cup I}$; then update the state of the algorithm by appending $(A_{\ell+1}, B_{\ell+1},d_{\ell+1})$ to $\cL$ and by incrementing $\ell$ by
    one.
  \item \emph{(Collapse phase)}
  Compute the canonical decomposition $\{I_0,\dots,I_\ell\}$ of $I$, and the corresponding canonical solution $W$.
  
  While  $\exists t : |I_{t}|\geq \mu|P_t|$: 
    \begin{enumerate}
      \item Choose the smallest such $t$.
      
     { \hspace*{-50pt} { \footnotesize\em //We refer to the following steps as collapsing layer $t$.}}
     \item Compute optimal solution $X$ to $\FLOW{M}{P_{\leq t-1}}{A_{\leq t} \cup I_{\leq t-1}}$ whose paths are	disjoint from $W_t$. 
     
          { \hspace*{-50pt} { \footnotesize\em //We refer to the following step  as alternating along the paths of $W_t$.}}
     \item For each path $\Pi$ in $W_t$ that ends at a player $p_e$ with an edge $(p_e,R)\in I_t$
      
      \begin{enumerate}
      	\item Set $ M \gets M \setminus \{(p, \{f\}) \; | \; (f, p) \in \Pi\} \cup \{(p, \{f\}) \; | \; (p, f) \in \Pi\}$. 
      	\item Remove from $M$ and $B_t$ the edge containing the source of the path $\Pi$.
      	\item Add to $M$ some $\beta$-edge $(p_e,R')$, where $R'\subseteq R$ and $R'\cap\cR(M)=\emptyset$. 
      \end{enumerate}
      
      \item Set $I=I_0\cup\ldots\cup I_{t-1}$. For every edge $(p,R)\in A_t$, if $v_p(R\setminus \cR(M))\geq \tau/\beta$, then:
      \begin{itemize}
      	\item Remove $(p,R)$ from $A_t$ and remove those edges from $B_t$ that only block $(p,R)$ in $A_t$.
      	\item If $X$ contains a path that ends in $p$, insert $(p,R)$ in $I$.
      \end{itemize}  
      
      \item Discard $(A_i, B_i,d_i)$ from $\cL$ with $i>t$ and set $\ell = t$.
    \end{enumerate}
\end{enumerate}
Output the matching $M$ that also matches $p_0$.
}
\end{small}
%\end{center}
\caption{\textsc{Combinatorial Augmenting Algorithm}}
\label{alg2}
\end{algorithm}

%%% Local Variables:
%%% mode: latex
%%% TeX-master: "main"
%%% End:

The combinatorial  algorithm behind the proof of Theorem \ref{thm:main} is stated as Algorithm \ref{alg2}. We remark that the computation of canonical decompositions and solutions to flow networks that are carried out in Steps 3 and 3.b respectively can be carried out in polynomial time; this fact is proved in Section \ref{subsec:single-iter}.

Similar to Algorithm~\ref{alg1}, Algorithm \ref{alg2} preserves the following invariants:
	
	\begin{enumerate}
		\item\label{combinv1} For $i=0,\ldots, \ell$, $A_i$ is a set of thin $\alpha$-edges and each $\alpha$-edge $(p,R) \in A_i$ has $R \cap
\res(A_{\leq i} \cup B_{\leq i-1}\cup I \setminus \{(p,R)\}) = \emptyset$ (its resources are not shared with edges from earlier iterations,
edges in $A_i$, or edges in $I$).
		
    \item\label{combinv2}  For any edge $(p,R)\in I$, it holds that $R \cap \res(A_{\leq \ell} \cup I \setminus \{(p,R)\}) = \emptyset$ and
      $v_p(R\setminus\cR(M))\geq \tau/\beta$.\footnote{We note that this invariant says that each edge $(p, R) \in I$ has a subset $R' \subseteq R$ so that $v_p(R') \geq \tau/\beta$ and $R' \cap \cR(A_{\leq \ell} \cup B_{\leq \ell}  \cup I \setminus \{(p, R)\}) = \emptyset$ since $B_{\leq \ell} \subseteq M$. In other words, the resources of $(p, R')$ are disjoint from all other resources in $\cL$.}
		
% 		\item\label{combinv3} For $i=0,\ldots, \ell$, $B_i =\{(p,R) \in M : \mbox{ $(p,R)$ is blocking an edge  in $A_i$}\}$
% 		%$B_i = \{(p,R) \in M : \mbox{ there exists $(p',R') \in A_i$  such that $R \cap R' \neq \emptyset$}\}$
% 		contains those thin edges of $M$ that blocks edges in $A_i$, where we say that an edge $(p, R) \in M$ blocks an edge $(p',
% R')$ if $R \cap R' \neq \emptyset$.

		\item\label{inv5} Given a canonical decomposition  $\{I_0,\dots,I_\ell\}$ of $I$,
		for $i=0,\ldots, \ell$ it holds that $|I_i| < \mu |P_{i}|$.

	\end{enumerate}
	
%	Notice the similarities between the above invariants and the invariants that were preserved by our simpler Algorithm \ref{alg1}.
The similarities  between these invariants and those of the simpler algorithm follow from the
same basic ideas. However, since Algorithm \ref{alg2} is more
involved, its analysis requires more invariants that we present in the subsequent sections. 
%additional invariants in Section~\ref{subsec:invariants}.

Before proceeding with analyzing Algorithm \ref{alg2}, we explain its steps in more detail and why the algorithm satisfies the above invariants.  The algorithm begins with a partial
matching $M$ and a player $p_0$ that we wish to include in our partial matching. Furthermore, as pointed out earlier, we make sure that $M$ contains a maximum matching between fat resources and players. Every iteration of our algorithm
involves two main phases: the build phase, and the collapse phase.

During the build phase of layer $\ell+1$, the algorithm finds thin $\alpha$-edges for the players in $P_\ell$ that we then insert into
either $I$ (if the $\alpha$-edge contains sufficient resources that do not appear in $M$) or to $A_{\ell+1}$. 
By the design of Algorithm \ref{alg2}, any edge that is inserted into $A_{\ell+1}$ will be disjoint from edges in
$A_{\leq \ell+1} \cup B_{\leq \ell}\cup I$; the same holds for any edge $(p,R)$ that is inserted into $I$, while in addition we have
$v_p(R\setminus\cR(M))\geq \tau/\beta$.
Therefore, the first two invariants are preserved during the build phase.
 
 Furthermore, edges inserted into $A_{\ell+1}$ or $I$ need to either contain a player from $P_{\leq \ell}$, or to be the final edge in an
alternating path that includes fat edges originating at a player in $P_{\leq \ell}$. Even though we will not store such alternating
paths explicitly, it is required that after we insert any such thin $\alpha$-edge into $A_{\ell+1}$ and $I$, the value of the flow network
$\DP{M}{P_{\leq \ell}}{A_{\leq \ell+1}\cup I}$ increases; this will ensure that there are enough disjoint paths of fat edges to
permit the inclusion of all such thin edges into our partial matching $M$. 
% Finally, after the insertion of edges into $A_{\ell+1}$
% and $I$, we define the set of blocking edges $B_{\ell+1}$ in such a way, that the third invariant is also upheld. Therefore, after the
% build
% phase, the first three invariants are upheld.
 
After the algorithm has finished the build phase, it proceeds to the collapse phase. The condition $\exists t: |I_t| \geq \mu |P_t|$ of the while-loop guarantees that the third invariant is satisfied once the collapse phase terminates (since the cardinality of $I_i$ always equals $\DP{M}{P_{\leq i}}{I}$ no matter the chosen canonical decomposition). We now describe this phase in more detail and show that it maintains a valid matching and that it does not introduce any violations of the first two invariants.  

The first step of the collapse phase is to  compute a canonical decomposition
of $I$, and a corresponding canonical solution $W$. Now suppose that we
have $I_t \geq \mu |P_t|$ and that the algorithm collapses layer $t$. We refer
to the edges of $I_t$ as "immediately addable" as they have enough free
resources (by Invariant~\ref{combinv2}) to be added to the matching. Indeed,
these are the edges we will insert into our partial matching, using the paths
of $W_t$. Specifically, for each path $\Pi$ of $W_t$, the algorithm proceeds as
follows. By definition of the sources and the sinks,  $\Pi$ is a path that starts with a player $p_s \in P_t$ and
ends with a player $p_e$ such that $(p_e, R) \in I_t$. Between $p_s$ and
$p_e$, the path alternates between fat edges that belong to $M$  and fat  edges we want to  insert into $M$, i.e., $\Pi = (p_s = p_1 , f_1, p_2, f_2,  \ldots, p_k, f_k, p_{k+1} = p_e)$ where $p_s$ is interested in $f_1$, $p_{k+1}$ is currently assigned $f_k$, and $p_i$ is currently assigned  $f_{i-1}$ and interested in $f_{i}$ for $i=2, \ldots, k$. 
To update the matching, we find a $\beta$-edge $(p_e, R')$ with $R' \subseteq R$ that is disjoint from the resources of matching $M$ (guaranteed to exist by the second invariant) and we let $(p_s, R_s)$ denote the edge in $B_t\subseteq M$ incident to player $p_s$. Step 3.c  now updates the matching by inserting $(p_e, R')$ and $(p_s, f_1), (p_2, f_2), \ldots, (p_k, f_k)$ to the matching while removing $(p_s, R_s)$ and $(p_2, f_1), (p_3, f_2), \ldots, (p_{t}, f_k)$.  This process is called
{\em alternating along path $\Pi$}. 

%We remark that all assigned fat resources remained assigned by this process and therefore we maintain a matching with a maximum number of fat items assigned to players. In addition, the process maintains that all players that were previously matched remain matched and thus the algorithm extends the given matching by also matching player $p_0$ once it terminates.
 
 As a result, some of the resources of edges in $A_{t}$ are freed up, and we move those edges of $A_t$ that now have $\tau/\beta$ free resources to
$I$ (Step 3.d). {Finally}, we discard all layers above the one we collapsed. 
Let us now see why our first invariant is upheld after the collapse phase. When we collapse layer $t$, we might remove edges
from $A_t$, we discard all $A_{t'}$ for $t'>t$ and we preserve $A_{t'}$ for $t'<t$. Since the first invariant was upheld before
the collapse phase, for any $t'\leq t$ there were no edges in $A_{t'}$ that intersected any edge in $A_{\leq t'}$, $B_{\leq t'-1}$ or
$I_0\cup\ldots I_{t-1}$. Furthermore, since any edge that was inserted into $I$ during Step 3.d previously belonged to $A_t$, no edge
inserted into $I$ will intersect any edge in $A_{\leq t}\cup B_{\leq t-1}\cup I_0\cup\ldots I_{t-1}$. Therefore, after the collapse phase,
for any $t'\leq t$ every edge in $A_{t'}$ is disjoint from edges in $A_{\leq t'}\cup B_{\leq t'-1}\cup I$, and the first invariant holds.

After the collapse phase, $I$ contains the edges that belonged to $I_0\cup\ldots I_{t-1}$ (call them {\em old} edges), plus the edges that
were inserted during Step 3.d (call them {\em new} edges). Concerning any old edge $e$, since the
second invariant held before the collapse phase, and since during the collapse phase for any $t'\leq t$ we introduce no new edges into
$A_{t'}$, the resources of $e$ continue to be disjoint from the resources of $A_{\leq t}$ and the old edges. Moreover, $v_p(\cR(e)\setminus \cR(M))$ is still at least  $\tau/\beta$ since the resources of the edges added to the matching during the collapse phase are disjoint from $\cR(e)$, where we use that the second invariant held before this iteration, i.e., that the resources of edges in $I$ are disjoint.  Hence, to verify the second
invariant it remains to verify that any new edge $(p,R)$ has $v_p(R\setminus \cR (M))\geq \tau/\beta$ (follows immediately from Step 3.d)
and that its resources are disjoint from the resources of all old and other new edges and edges in $A_{\leq t}$; but this follows directly from the fact
that any new edge belonged to $A_t$ before the collapse phase and the fact that the first invariant held before the collapse phase. Hence,
the second invariant is satisfied after the collapse phase.

% Finally, concerning the third invariant, we note that $B_t$ changes during Steps 3.c and 3.d, $B_{t'}$ is discarded for $t'>t$ and
% $B_{t'}$
% does not change for $t'<t$; since $A_{t'}$ also does not change for $t'<t$, the invariant continues to hold for edges in $B_{t'}$ for
% $t'\neq t$. Let us now focus on $B_t$: during Step 3.c, we remove edges from $B_t$, and since we also remove them from $M$, the invariant
% holds. Finally, during Step 3.d, we remove edges from $A_t$, and their corresponding blocking edges from $B_t$; hence, any edge that
% remains in $B_t$ after Step 3.d is blocking some edge from $A_t$. Therefore, the third invariant continues to hold after the collapse
% phase.

Now, let us see why the output of Algorithm \ref{alg2} is a partial matching that matches player $p_0$. Observe
that we only update our partial matching during Step 3.c and, as explained above, we alternate along all paths in $W_t$  during this step. As these paths are vertex-disjoint and the edges in $I$ have disjoint resources (by the second invariant), these updates do not interfere with each other. Moreover, note that when we alternate along a path all previously matched players remain matched (albeit to new edges) and, in addition, all fat resources remain matched. This means that our algorithm maintains a matching of the players that were matched by the input and that our matching remains one that  maximizes the number of assigned fat resources.  By iterating until an edge
that contains $p_0$  is inserted into $M$, it follows that when Algorithm 2 terminates, the output will  be a valid matching that also matches $p_0$ in addition to the players that were matched by the original matching that was given as input.

Our running time analysis of Algorithm \ref{alg2} is carried out in the following sections. Specifically, we begin by analyzing the running
time of a single iteration of our augmenting algorithm in Section \ref{subsec:single-iter}. Then, we proceed to state certain
additional invariants,
and prove that they are upheld by Algorithm \ref{alg2} in Section \ref{subsec:invariants}.
 Finally, using these invariants, we will prove that the total number of iterations executed by Algorithm \ref{alg2} is polynomial in
Section \ref{subsec:bound-iter}.

\subsection{Running Time Analysis of a Single Iteration}\label{subsec:single-iter}
In this section, we prove that the running time of a single iteration is
polynomial. We begin by studying the build phase. In this phase, in each
iteration of the while-loop, we consider those players $p$ and resources $R$
such that $(p,R)$ is a thin $\alpha$-edge satisfying $R \cap \cR(A_{\leq
\ell+1} \cup B_{\leq \ell} \cup I) = \emptyset$. We then check whether adding
$p$ as a sink to our flow network strictly increases its value, i.e., the
number of disjoint paths from the sources in $P_{\leq \ell}$ to the sinks in
$A_{\leq \ell +1} \cup I \cup \{(p,R)\}$.  Both these operations can be done in polynomial time as  (1) verifying whether such a set $R$ exists for a player $p$ just amounts to calculating the total value of the resources $p$ is interested in that currently are not in the other relevant edges, and as (2)  verifying whether the flow network increases its value reduces to  a standard maximum flow problem. 

Next, we study the collapse operation. Here, we have two non-trivial operations: computing a canonical decomposition (Step~3 of Algorithm \ref{alg2}) and Step 3.b of Algorithm \ref{alg2}.
%Regarding the computation of a canonical decomposition, we proceed to prove that it can be done in polynomial time:
\begin{lemma}\label{lem:canonicaldecomposition} Given a state 
	$(M, \ell, \cL, I)$ of the algorithm, we can find a canonical decomposition  of $I$ and the corresponding canonical solution in polynomial time.
\end{lemma}
\begin{proof}
  	We shall construct an optimal solution $W$ to the flow network
        $\FLOW{M}{P_{\leq \ell}}{I}$ with
  sources $P_{\leq \ell}$ and sinks $I$ iteratively. Compute
    the maximum flow $X^{(0)}$ in the network
    $\FLOW{M}{P_{\leq 0}}{I}$. Let $Q_0 \subseteq P_0$ be the set of
    sources appearing in the flow solution $X^{(0)}$. Now observe that
    this solution $X^{(0)}$ is also a valid flow in the network
    $\FLOW{M}{P_{\leq 1}}{I}$. Therefore, by using an augmenting flow
    algorithm, we can augment the flow $X^{(0)}$ to a maximum flow
    $X^{(1)}$ in the network $\FLOW{M}{P_{\leq 1}}{I}$. Let
    $Q_1 \subseteq P_1$ be the set of additional sources appearing in the flow
    solution $X^{(1)}$. We use here an important property of the flow augmentation
    process, which states that the set of sources in $X^{(1)}$ is precisely the
    disjoint union $Q_0 \cup Q_1$ (see, for
    example,~\cite{schrijver2002combinatorial}). In other words, a
    vertex appearing as a source of a flow path in a solution
    continues to be present as a source of a flow path after an
    augmentation step. Continuing this process, we end up with a flow
    solution $X^{(\ell)}$ in the network $\FLOW{M}{P_{\leq \ell}}{I}$. Define $W_i$ to be the flow paths in
    $X^{(\ell)}$ that serve the sources $Q_i \subseteq P_i$  for each
    $i=0,\dots,\ell$. Additionally, let $I_i \subseteq I$ denote the sinks of $W_i$. 

  By construction, $|I_{\leq i}| = \DP{M}{P_{\leq i}}{I_{\leq
      i}}$. Further, if
  $\DP{M}{P_{\leq i}}{I_{\leq i}} < \DP{M}{P_{\leq i}}{I}$ then this
  implies that $X^{(i)}$ is not a maximum flow in
  $\FLOW{M}{P_{\leq i}}{I}$, and therefore can be augmented by one,
  contradicting the definition of $X^{(i)}$.

  The flow paths $W_0, W_1, \ldots, W_\ell$ collectively form the flow
  solution $X^{(\ell)}$ which is an optimal solution to
  $\FLOW{M}{P_{\leq \ell}}{I}$. Thus, $\{I_0,\dots,I_\ell\}$ forms a
  canonical decomposition (with the corresponding canonical solution $W_0, \ldots, W_\ell$). It is also clear that the process outlined
  above to realize this decomposition runs in polynomial time as the
  encountered flow networks have unit capacities.
% 	We shall construct an  optimal solution $W$ to the flow network with sources $P_{\leq \ell}$ and sinks $I$ iteratively. We begin by calculating
% 	an optimal solution when only considering the sources $P_0$, and then we augment it to obtain an optimal
% 	solution to the case with sources $P_{\leq 1}$, and so on 
% 	until we have an optimal solution with sources $P_{\leq \ell}$; the fact that we can always perform such an augmentation iteratively, until we reach the maximum value of the flow netwrok, follows from the fact that all sets $I_{\leq i}$ during this process correspond to an independent set of a matroid, and therefore, due to the {\em matroid augmentation property}, we can always augment any $I_{\leq i}$ until we find a maximum cardinality independent set. 
% 	
% 	In that optimal solution, define $I_i$ to contain those edges in $I$ that contain players in $P_i$.
% 	The properties of a canonical decomposition are now easy to prove. By construction, $W$ is an optimal solution of $\FLOW{M}{P_{\leq l}}{I}$. By the definition of $I_i$, $|I_{\leq i}| = \DP{M}{P_{\leq i}}{I_{\leq i}}$ for each $i=0,\dots,\ell$. The fact that $\DP{M}{P_{\leq i}}{I_{\leq i}}=\DP{M}{P_{\leq i}}{I}$ for each $i=0, 1, \ldots,\ell$ follows because $\DP{M}{P_{\leq i}}{I_{\leq i}} < \DP{M}{P_{\leq i}}{I}$ would contradict the fact that $|I_{\leq i}|$ was the value of the maximum flow in the network $\FLOW{M}{P_{\leq i}}{I}.$
\end{proof}

Next, we prove that Step 3.b can be executed in polynomial time:
\begin{lemma}
	\label{lem:disjointpaths}
	Consider a state $(M, \ell, \cL, I)$ of the algorithm  and a canonical decomposition $\{I_0,
	I_1, \ldots, I_\ell\}$ of $I$ together with the canonical solution  $W$. For $i=0, \ldots, \ell,$ let $W_i$ be the $|I_i|$ paths that go
	from the players in $Q_i \subseteq P_{i}$ to sinks in $I_i$. Then, for $i=0, 1, \ldots,
	\ell-1$, we can find in polynomial time an optimal solution $X$ to
	$\FLOW{M}{P_{\leq i}}{A_{\leq i+1}\cup I_{\leq i}}$ that is also an optimal solution to
	$\FLOW{M}{P_{\leq i}}{A_{\leq i+1}\cup I}$  whose paths are disjoint from the paths in
	$W_{i+1}$ and additionally uses all the sinks in $I_{\leq i}$. 
\end{lemma}
\begin{proof}
	Consider a fixed $i$. We shall form an optimal solution $X$ to $\FLOW{M}{P_{\leq i}}{A_{\leq
			i+1}\cup I_{\leq i}}$ that is also an optimal solution to $\FLOW{M}{P_{\leq i}}{A_{\leq
			i+1}\cup I}$ and its paths are disjoint from the paths in $W_{i+1}$ and uses all the sinks in
	$I_{\leq i}$. The initial
	solution will be the set of unit flow paths $W_{\leq i}$ from the canonical solution $W$ which has
	cardinality $|I_{\leq i}|$. We now augment this solution using augmenting paths to the set of
	sinks $A_{\leq i+1}$. Note that throughout this execution each vertex in $I_{\leq i}$ will be used
	as a sink by some path and therefore $X$ will use all these sinks. Further, the procedure to
	calculate $X$ clearly runs in polynomial time. We shall
	now verify the remaining properties of $X$. First, suppose towards contradiction that some iteration used an augmenting path $P$ intersecting
	a path in $W_{i+1}$. However, this would imply that there exists an augmenting path that uses a
	sink in $I_{i+1}$. We could then increase the set of disjoint paths from players in $P_{\leq i}$ to
	sinks in $I$ to be greater than $I_{\leq i}$ which contradicts the property $\DP{M}{P_{\leq
			i}}{I_{\leq i}} = \DP{M}{P_{\leq i}}{I}$ of the canonical decomposition.  Similarly, suppose
	$X$ is not an optimal solution to $\FLOW{M}{P_{\leq i}}{A_{\leq
			i+1}\cup I}$. Then there exists an augmenting path to an edge in $I\setminus I_{\leq i}$
	which again contradicts the property $\DP{M}{P_{\leq
			i}}{I_{\leq i}} = \DP{M}{P_{\leq i}}{I}$ of the canonical decomposition.
\end{proof}

Finally, since during a collapse operation we can collapse at most $|\cP|$ layers, it follows that any iteration of Algorithm \ref{alg2} terminates in polynomial time.

\subsection{Additional Invariants of Combinatorial Algorithm}\label{subsec:invariants}
In Section \ref{subsec:combalgo-description}, we listed three invariants
Algorithm \ref{alg2} preserves that are similar to the simpler algorithm. We
argued why they hold, and how these invariants imply that the output of our
algorithm is an extended partial matching. In this section, we list two new
invariants that will facilitate our polynomial running time proof.

\begin{lemma}
  \label{lem:ainv}
  At
the beginning of each iteration:
\begin{enumerate}[label=(\alph*)]
    \item \label{ainv:1} $\DP{M}{P_{\leq \ell}}{I} = |I|$.
    \item \label{ainv:2} $\DP{M}{P_{\leq i-1}}{A_{\leq i}\cup I} \geq d_i$ for each $i=1,\dots,\ell.$
  \end{enumerate}
\end{lemma}
\begin{proof}
  We prove the lemma by induction on the number of times the iterative step has been executed.
We observe that both invariants trivially hold before the first execution of the iterative step. Assume that they are true before the $r$-th execution of the iterative step. We now verify them before the $r+1$-th iterative step. We actually prove the stronger statement that they hold after the build phase and after each iteration of the collapse phase. 

\textbf{\ref{ainv:1} and~\ref{ainv:2} hold after the build phase.} Let $L_{\ell+1}$ denote the layer that was constructed during the
build phase. We start by verifying~\ref{ainv:1}. If no edge is added to $I$ during this phase then $|I| \geq \DP{M}{P_{\leq \ell+1}}{I} \geq \DP{M}{P_{\leq \ell}}{I} = |I|.$ Suppose that $a_1,\dots,a_k$ were the edges added to the set $I$ in that order. When edge $a_i$ was added to the set $I$,  from the definition of Step 2 of Algorithm \ref{alg2} we have that 
\[
  \DP{M}{P_{\leq \ell}}{A_{\leq \ell}\cup I \cup \{a_1,\dots,a_{i-1}\} \cup \{a_i\}} > \DP{M}{P_{\leq \ell}}{A_{\leq \ell}\cup I \cup \{a_1,\dots,a_{i-1}\}},
\]
 which then implies that 
 \[
   \DP{M}{P_{\leq \ell}}{I \cup \{a_1,\dots,a_{i-1}\} \cup \{a_i\}} > \DP{M}{P_{\leq \ell}}{I \cup \{a_1,\dots,a_{i-1}\}}
 \] 
 To see this implication, observe that the first inequality implies that, for any flow in $\FLOW{M}{P_{\leq \ell}}{A_{\leq \ell}\cup I \cup \{a_1,\dots,a_{i-1}\}}$ (and hence, for any flow in $\FLOW{M}{P_{\leq \ell}}{I \cup \{a_1,\dots,a_{i-1}\}}$), there exists an augmenting path towards sink $a_i$.
  Along with the induction hypothesis, these inequalities imply that \[\DP{M}{P_{\leq \ell+1}}{I\cup \{a_1,\dots,a_k\}} \geq \DP{M}{P_{\leq \ell}}{I\cup \{a_1,\dots,a_k\}} = |I| + k = |I\cup \{a_1,\dots,a_k\}|.\]

  For~\ref{ainv:2}, the inequality for $i = \ell+1$ holds by the definition of $d_{\ell+1}$ during this phase. The remaining inequalities follow from the induction hypothesis since none of $M, P_{\leq \ell}$ and $A_{\leq \ell}$ were altered during this phase and no elements from $I$ were discarded.

\textbf{\ref{ainv:1} and~\ref{ainv:2} hold after each iteration of the collapse phase.} If no layer is collapsed (i.e., there is no $I_t$ satisfying the condition of the while-loop) then there is nothing to prove. Now  let $t$ denote the index  of the layer that is
collapsed. Let $(M, \ell, \{L_0,\ldots,L_{t'}\}, I)$ denote the
state of the algorithm before collapsing layer $t$  that satisfy~\ref{ainv:1} and~\ref{ainv:2} ($t' \geq t$ and $t' = \ell+1$ if this is the first iteration of Step~3). Let $I'$ denote $I_0\cup \dots \cup I_{t-1}\cup \{a_1,\dots,a_k\}$
where $a_1,\dots,a_k$ are the edges added to $I$ in Step 3.d
of the collapse phase and let $M'$ denote the partial
matching after Step 3.c of the collapse phase. We have that~\ref{ainv:1}, $\DP{M'}{P_{\leq t}}{I'} = |I'|$, now follows
from Lemma~\ref{lem:disjointpaths}. Indeed, the solution $X$ used 
all the sinks in $I_{0} \cup \ldots I_{t-1} \cup \{a_1,\dots,a_k\}$ which equals $I'$; and these
paths form a solution to $\FLOW{M'}{P_{\leq t}}{I'}$ as they are disjoint from the paths in $W_t$.
% and that $I_0,\dots,I_t$ is a canonical decomposition of
%$I$. 
Notice that we do not use the induction hypothesis in this case, i.e., that $(M, \ell, \{L_0,\ldots,L_{t'}\}, I)$ satisfied~\ref{ainv:1} and~\ref{ainv:2}.  

For~\ref{ainv:2}, we need to verify inequalities for $i=1,\dots,t$. When $i < t$, none
of the sets $A_i$ were altered during this iterative step.  Further, although $M$ and $I$ changes during the collapse phase, by Lemma~\ref{lem:disjointpaths}
and the definition of Step 3
this change cannot reduce the number of disjoint paths from $P_{\leq i-1}$ to $A_{\leq i}\cup
I$ and therefore~\ref{ainv:2} remains true  by the induction hypothesis. Indeed, the selection of $X$ in Step~3.b is done so as to make sure that the update of the matching along the alternating paths in $W_t$ does not interfere with an optimal solution to the flow network with sources $P_{\leq i-1}$ and sinks $A_{\leq i} \cup I$. For $i = t$, the claim again follows since the number of disjoint paths from $P_{\leq t-1}$ to $A_{\leq t}\cup I$ cannot 
reduce because of Step 3.d in the algorithm that maintains $X$ as a feasible
solution by the same arguments as for~\ref{ainv:1}.
\end{proof}

\subsection{Bound on the Total Number of Iterations}\label{subsec:bound-iter}

In this final section, we will use the above invariants to show that our augmenting algorithm performs a polynomial number of iterations,
assuming $CLP(\tau)$ is feasible. We start with two lemmas that show that $d_i$ cannot be too small. The first holds in general and the second holds if $CLP(\tau)$ is feasible.
\begin{lemma}\label{lem:dproperty}
	At the beginning of each iteration, we have $d_i \geq |A_{\leq i}|$ for every $i=0,\dots,\ell$.
\end{lemma}

\begin{proof}
	We prove this by induction on the variable $r \geq 0$ that counts the number of times the
	iterative step has been executed. For $r = 0$ the statement is trivial. Suppose that it is true
	for $r \geq 0$.   We shall show that it  holds before the $r+1$-th iterative step. If the iteration	collapses a layer, then no new layer was added, and since $d_i$'s remain unchanged and $A_{\leq i}$ may only decrease, the statement is true in this case. 
	
	Now, suppose that no layer was collapsed in this
	iteration and let $L_{\ell+1}=(A_{\ell+1},B_{\ell+1},d_{\ell+1})$ be the newly constructed
  layer in this phase. Again, we have $d_i \geq |A_i|$ for $i= 0, \dots, \ell$ since none of these quantities are changed by the build phase. Let us now verify that $d_{\ell+1} \geq A_{\ell+1}$. Let  $A_{\ell+1} = \{a_1,\dots,a_k\}$ denote the set of 
	edges added to $A_{\ell+1}$ indexed by the order in which they were added. When edge
	$a_i$ was added to the set $A_{\ell+1}$, according to Step 2 of Algorithm \ref{alg2}, we have that \[\DP{M}{P_{\leq \ell}}{A_{\leq \ell}\cup I\cup \{a_1,\dots,a_{i-1}\} \cup \{a_i\}} > \DP{M}{P_{\leq \ell}}{A_{\leq \ell}\cup I \cup
		\{a_1,\dots,a_{i-1}\}}.\] 
	%which implies that \[\DP{M}{P_{\leq \ell}}{A_{\leq \ell} \cup
	%   \{a_1,\dots,a_{i-1}\} + a_i} > \DP{M}{P_{\leq \ell}}{A_{\leq \ell} \cup \{a_1,\dots,a_{i-1}\}}\]
	% since the first inequality implies that there exists an augmenting path with $a_i$ as the
	% sink. 
    Using~\ref{ainv:2} of Lemma~\ref{lem:ainv}  and the induction hypothesis,
	\[ \DP{M}{P_{\leq \ell-1}}{A_{\leq \ell}\cup I} \geq d_\ell \geq |A_{\leq \ell}|. \]
	Using the previous inequalities,
	\[ d_{\ell+1} = \DP{M}{P_{\leq \ell}}{A_{\leq \ell+1}\cup I} \geq
	|A_{\leq \ell}| + k \geq |A_{\leq \ell+1}|.\] 
  %The remaining inequalities continue to hold since, for any $0 \leq i \leq \textcolor{red}{\ell}$, $d_i$ does not change after layer $L_i$ is built. Therefore, the inequalities continue to hold before the $r+1$-th iterative step as well.
\end{proof}

\begin{lemma}
	\label{lem:manyaddableedges}
	Assuming $CLP(\tau)$ is feasible, at the beginning of each iteration
	\[
	\DP{M}{P_{\leq i-1}}{A_{\leq i}\cup I} \geq d_{i} \geq  \gamma |P_{\leq i-1}| , \; \mbox{where } \; \gamma = \frac13(\sqrt{10}-2),
	\]
	for every $i = 1, \dots,\ell.$
\end{lemma}

\begin{remark}
	\label{rem:abort}
	The above condition is the only one that needs to be satisfied for the algorithm to run in
	polynomial time. Therefore, in a binary search, the algorithm can abort if the above condition is
	violated at some time, since that violation would imply that the Configuration-LP is infeasible; otherwise it will terminate in polynomial time.
\end{remark}

\begin{proof}
  We will prove that $d_i \geq \gamma|P_{\leq i-1}|$ for $i=1,\dots,\ell$ as Lemma~\ref{lem:ainv}\ref{ainv:2} then implies the claim. Notice that $d_i$ is defined only at the time when
	layer $L_i$ is created and not altered thereafter. So it suffices to verify that: Assuming $d_i \geq
	\gamma|P_{\leq i-1}|$ for $i=1,\dots,\ell$, then for the newly constructed layer $L_{\ell+1}$, $d_{\ell +
		1} \geq \gamma|P_{\leq \ell}|$ also.
	
	Suppose towards contradiction that $L_{\ell+1}$ is a newly constructed layer (and that no layer was collapsed), such that
	\[
	d_{\ell + 1} = \DP{M}{P_{\leq \ell}}{A_{\leq \ell+1}\cup I} < \gamma |P_{\leq \ell}|.
	\]
	
  Then, since no layer was collapsed at Step 3 of Algorithm \ref{alg2}, we have that  $|I_i| < \mu |P_i|$ for $i=0, \ldots, \ell$, where $\{I_0,\ldots,I_\ell\}$ is the  canonical decomposition of $I$ considered by the algorithm. {Together} with Lemma~\ref{lem:ainv}\ref{ainv:1}, this implies
	\begin{align*}
	|I| = \DP{M}{P_{\leq \ell}}{I} &< \mu |P_{\leq \ell}|.
	\end{align*}
	Moreover, by Lemma~\ref{lem:dproperty}  we have
	\begin{align*}
	|A_{\leq \ell+1}| \leq d_{\ell+1} =  \DP{M}{P_{\leq \ell}}{A_{\leq \ell+1}\cup I} &< \gamma|P_{\leq \ell}|.
	\end{align*}
	Hence, we have that $|A_{\leq \ell+1}\cup I | < (\mu+\gamma)|P_{\leq \ell}|.$ 
	
	The rest of the proof
	is devoted to showing that this causes the dual of the $CLP(\tau)$ to become unbounded which leads
	to the required contradiction by weak duality. That is, we can then conclude that if $CLP(\tau)$ is
	feasible then $d_{\ell+1} \geq \gamma |P_{\leq \ell}|$.
	
	Consider the flow network $\FLOW{M}{P_{\leq \ell}}{A_{\leq \ell+1}\cup I \cup Z}$ with $P_{\leq \ell}$ as the set of
	sources and $A_{\leq \ell+1}\cup I\cup Z$ as the collection of sinks where,
	 \[Z := \{p \in	\mathcal{P} \; | \; \text{$\exists R \subseteq
             \mathcal{R} \; : \; R \cap \mathcal{R}(A_{\leq
	 		\ell+1}\cup I \cup B_{\leq \ell}) = \emptyset$
                      and $v_p(R)\geq \tau/\alpha$} \}.\] 
	 Since, during the construction of layer $\ell+1$ we could not insert any more edges into $A_{\ell+1}$ and $I$, the maximum number of vertex disjoint paths from $P_{\leq \ell}$ to the sinks equals $\DP{M}{P_{\leq \ell}}{A_{\leq \ell+1} \cup I}$ which, by assumption, is less than $\gamma |P_{\leq
		\ell}|$. Therefore, by Menger's theorem there exists a set $K \subseteq V$ of vertices of
	cardinality less than $\gamma |P_{\leq \ell}|$ such that, if we remove $K$ from $H_M$, the sources $P_{\leq \ell}	\setminus K$ and the sinks are disconnected, i.e., no sink is reachable from any source in $P_{\leq
		\ell} \setminus K$. We now claim that we can always choose such a vertex cut so that it is a
	subset of the players.
	
	\begin{claim}
		There exists a vertex cut $K \subseteq \cP$ separating $P_{\leq \ell}\setminus K$ from the sinks of cardinality less than $\gamma|P_{\leq {\ell}}|$.
	\end{claim}
	
	\begin{proof}
		Take any minimum cardinality vertex cut $K$ separating $P_{\leq \ell}\setminus K$ from the sinks. We already saw that $|K| < \gamma|P_{\leq \ell}|.$ Observe that every fat resource that is reachable
		from $P_{\leq \ell} \setminus K$ must have outdegree exactly one in $H_M$. It cannot be more than one since
		$M$ is a collection of disjoint edges, and it cannot be zero since we could then increase
		the number of fat edges in  $M$ which contradicts that we started with a partial
		matching that maximized the number of fat edges. Therefore in the vertex
		cut $K$, if there are vertices corresponding to fat resources, we can replace each fat
		resource with the unique player to which it has an outgoing arc to, to obtain  another vertex cut also of the same cardinality that contains only vertices corresponding to players.
	\end{proof}
	
	Now call the induced subgraph of $H_M-K$ on the vertices that are reachable from $P_{\leq
		\ell} \setminus K$ as $H'$. Note that by the definition of $K$, $H'$ will not contain any sinks. Using $H'$ we define the assignment of values to the dual variables in the dual of $CLP(\tau)$ as
	follows:

	\[
	\begin{array}{lcl}
	y_i &:=&
	\begin{cases}
	\left(1-1/\alpha \right) &\mbox{if } \text{player $i$ is in $H'$,} \\
	0 &\mbox{otherwise,}
	\end{cases} \\ \\
	z_j &:=&
	\begin{cases}
	v_j/\tau &\mbox{if } \text{$j$ is a thin resource that appears in $A_{\leq \ell+1}\cup I \cup B_{\leq
			\ell}$, } \\
	\left(1-1/\alpha \right) &\mbox{if } \text{$j$ is a fat resource in $H'$,} \\
	0 &\mbox{otherwise. }
	\end{cases}
	
	\end{array}
	\]
	
	We first verify that the above assignment is feasible. Since all the dual variables are non-negative we only need to verify that $y_i \leq \sum_{j \in C} z_j$ for every $i \in \cP$ and $C \in \cC(i, \tau).$ Consider a player $i$ that is given a positive $y_i$ value by the above assignment. Let $C \in \cC(i, \tau)$ be a configuration for player $i$ of value at least $\tau$; we will call $C$ thin if it only contains thin resources, and fat otherwise. There are two cases we need to consider.
	
	\begin{itemize}
		\item[Case 1.] \textbf{C is a thin configuration.} Suppose that $\sum_{j \in C} z_j <
		(1-1/\alpha)$. Then, by our assignment of $z_j$ values, this implies that there exists a set $R
		\subseteq C$ such that $R$ is disjoint from the resources in $A_{\leq \ell+1}\cup I \cup B_{\leq \ell}$ and
		$\sum_{j \in R} v_j \geq \tau/\alpha.$ Together this
		contradicts the fact that $H'$ has no sinks since $i$ is then a sink (it is in $Z$).
		\item[Case 2.] \textbf{C is a fat configuration.} Let $j$ be a fat resource in $C$.
		Since $i$ was reachable in $H'$, all the sources in $H'$ are assigned thin edges in $M$ (which implies they have no incoming arcs), and $K$ is a subset of the players, it follows that $j$ is also present in $H'$.
		Thus, by our assignment, $z_j = 1-1/\alpha$.
	\end{itemize}
	
	Having proved that our assignment of $y_i$ and $z_j$ values constitutes a feasible solution to the
	dual of $CLP(\tau)$, we now compute the objective function value $\sum_i y_i - \sum_j z_j$ of the
	above assignment. To do so we adopt the following charging scheme: for each fat resource $j$ in
	$H'$, charge its $z_j$ value against the unique player $i$ such that the outgoing arc $(j, i)$
	belongs to $H'$. The charging scheme accounts for the $z_j$ values of all the fat resources except
	for the fat resources that are leaves in $H'$. There are at most $|K_1|$ such fat resources, where
	$K_1 \subseteq K$ is the set of players to which the uncharged fat items have an outgoing arc to. Moreover,
	note that $K_1$ only consists of players that are matched in $M$ by fat edges. Since $P_{\leq \ell}$ does not have any players matched by fat edges in $M$, no player in $K_2 := P_{\leq \ell}\cap K$ is present in $K_1$, i.e., $K_1 \cap K_2 = \emptyset.$ Finally, note that no player in $P_{\leq \ell} \setminus K = P_{\leq \ell} - K_2$ has been charged. Thus, considering all players in
	$\cP$ but only fat configurations, we have 
	\begin{align*}
	\sum_{i\in \cP} y_i - \sum_{j\in \cR_f} z_j &\geq 
	(1-1/\alpha)(|P_{\leq
		\ell}| - |K_2|) - (1-1/\alpha)|K_1| \\
	& =  
	(1-1/\alpha)\big(|P_{\leq \ell}|-(|K_1| + |K_2|)\big) \\
	& > (1-1/\alpha)(1-\gamma) |P_{\leq \ell}|.
	\end{align*}
	
	We now compute the total contribution of thin resources, i.e., $\sum_{j\in \cR_t} z_j$. The total
	value of thin resources from  the edges $A_{\leq \ell+1}$ and the edges
	$I$ is at most $(1/\alpha+1/\beta)|A_{\leq \ell +1}\cup I|$, due to the minimality of thin $\alpha$-edges. Besides the resources appearing in
	$A_{\leq \ell +1}\cup I$, the total value of resources appearing only in edges $B_{\leq \ell}$
	is at most $(1/\beta)(|B_{\leq \ell}|) < (1/\beta)(|P_{\leq \ell}|),$ by the minimality of $\beta$-edges. Indeed, if an edge in $B_\ell$ has more than $\tau/\beta$ resources not appearing in an	edge in $A_{\leq \ell +1} \cup I$ then those resources would form a thin $\beta$-edge which	contradicts its minimality.
	
	Using $|A_{\leq \ell+1}\cup I | < (\mu+\gamma)|P_{\leq \ell}|$ we have
	\[\sum_{i\in \cP} y_i - \sum_{j\in \cR} z_j > (1-\gamma)\left(1-\frac{1}{\alpha} \right) |P_{\leq \ell}| -
	(\mu + \gamma)\left(\frac1\alpha+\frac1\beta \right) |P_{\leq \ell}| - \frac{1}{\beta}|P_{\leq \ell}|.\]
	Recall that, given any feasible solution to the dual of $CLP(\tau)$, we can scale it by any positive number, and it will remain feasible; this will imply that if the optimum of the dual of $CLP(\tau)$ is positive, then the dual of $CLP(\tau)$ is unbounded. So, the dual of $CLP(\tau)$ is unbounded when \[(1-\gamma)\left(1-\frac{1}{\alpha} \right) -
	(\mu + \gamma)\left(\frac1\alpha+\frac1\beta \right) - \frac{1}{\beta} \geq 0
	\Leftrightarrow \gamma \leq \frac{\alpha \beta - (1+\mu)(\alpha+\beta)}{\alpha \beta + \alpha}.
	\]
	
	Recall that $\beta =2(3+\sqrt{10})+\epsilon$, $\alpha=2$, and $\mu = \epsilon/100$. For $\epsilon > 0$ the last inequality is equivalent to $206 \sqrt{10}+3 \epsilon\leq 676$, which is valid for $\epsilon \leq 1$.
\end{proof}

We now use the previous lemma to show that if we create a new layer then the number of players in that layer will increase rapidly. This will allow us to bound the number of
layers to be logarithmic and also to bound the running time.

\begin{lemma}[Exponential growth]\label{lem:expgrowthorcollapse}
	At each execution of the iterative step of the algorithm, we have
	\[|P_{i}| \geq \delta |P_{\leq i-1}|, \; \mbox{where } \;  \delta := \epsilon/100,
	%\frac{(\gamma-\mu)(\beta-\alpha)}{2\alpha} - 1,
	\]
	for each $i=1,\dots,\ell$.
\end{lemma}

\begin{proof}
	
	Suppose towards contradiction that the statement is false and let $t$ be the smallest index that
	violates it, i.e., $|P_{t}| < \delta |P_{\leq t-1}|$. Due to Invariant \ref{inv5}, $|I_i| < \mu|P_i|$ for $0 \leq i \leq t$.  Hence,
	\[|I_{\leq t}| < \mu|P_{\leq t}| < \mu(1+\delta) |P_{\leq t-1}|. \] 
	Further,
	$$
  |A_{\leq t}| + |I_{\leq t}| \geq \DP{M}{P_{\leq t-1}}{A_{\leq t}\cup I_{\leq t}} = \DP{M}{P_{\leq t-1}}{A_{\leq t}\cup I}  \geq \gamma |P_{\leq t-1}|,
	$$
  where the first inequality is trivial, the equality follows from the definition of canonical decompositions (Definition~\ref{def:cancomp}), and the last inequality follows
	from Lemma~\ref{lem:manyaddableedges}. This gives us \[ |A_{\leq t}|
	> \left(\gamma - \mu(1+\delta)\right)|P_{\leq t-1}|.\]
	
	We now obtain an upper bound on
  the total number of edges in $A_{\leq t}$ by counting the value of resources in each $A_i$ and $B_i$; observe that any thin $\beta$-edge
has resources of total value at most $2\tau/\beta$ due to minimality, while any thin $\alpha$-edge in $A_{\leq t}$ has resources of value at
least $\tau/\alpha - \tau/\beta$  that are blocked, i.e., appear in some edge in $B_{\leq t}$ (since otherwise this edge would be in $I$
instead of $A_{\leq t}$)\footnote{We remark that just as in the simpler algorithm, the set $B_i$ contains those edges that are
blocking the edges in $A_i$. This follows from the definition of the build phase and Steps 3.c and
3.d that remove edges from
$B_i$ when the matching has changed or when $A_i$ has changed. Furthermore, all edges in $A_i$ have all but at most $\tau/\beta$ resources
blocked. Otherwise, the edge is added to $I$ in the build phase and if resources have been freed up later, the edge is removed from
$A_i$ (and it may be added to $I$) during Step 3.c.}. Hence,
	\[|A_{i}|\left(\tau/\alpha-\tau/\beta\right) \leq
	|B_{i}|\left( 2\tau/\beta \right) \overset{\text{summing over
			$i$ and rearranging}}{\implies}  |A_{\leq t}| \leq |B_{\leq
		t}|\frac{2\alpha}{\beta-\alpha}. \] Since $|B_{\leq t}| < |P_{\leq
		t}|$ and $|P_{\leq t}| < (1+\delta) |P_{\leq t-1}|$ we have the
	bound \[|A_{\leq t}| < \frac{2\alpha}{\beta-\alpha}(1+\delta)|P_{\leq t-1}|.\]

	Therefore we will have a contradiction when 
	\[
	\frac{2\alpha}{\beta-\alpha}(1+\delta) \leq \gamma - (1+\delta)\mu.
	\]
	It can be verified that for any $\epsilon > 0$ the above inequality is equivalent
	to \[22400+6 \left(52+\sqrt{10}\right) \epsilon+3 \epsilon^2\leq 9400
	\sqrt{10},\] which is true for $\epsilon \in [0, 1]$ leading to the
	required contradiction.
\end{proof}

We are now ready to prove that our algorithm executes a polynomial number of iterations.
To do this, we define the signature vector $s := (s_0,\dots,s_\ell,\infty)$, where $$s_i := \lfloor \log_{1/(1-\mu)} \frac{|P_i|}{\delta^{i+1}} \rfloor$$ corresponding to the state $(M, \ell, \cL, I)$ of the algorithm. The signature vector changes as the algorithm executes; in fact, we prove that its lexicographic value always decreases:

	\begin{lemma}
		Across each iterative step, the lexicographic value of the signature vector decreases. Furthermore, the coordinates of the signature vector are always non-decreasing.
	\end{lemma}
	
	\begin{proof}
		We show this by induction as usual on the variable $r$ that counts the number of times the
		iterative step has been executed. The statement for $r = 0$ is immediate. Suppose it is true for
		$r \geq 0$. Let $s = (s_0,\dots,s_\ell,\infty)$ and $s' = (s'_0,\dots,s'_{\ell'},\infty)$ denote
		the signature vector at the beginning and at the end of the $(r+1)$-th iterative step. We consider
		two cases:
		\paragraph{No layer was collapsed} Let $L_{\ell + 1}$ be the newly constructed layer. In this
		case, $\ell' = \ell+1$. By Lemma~\ref{lem:expgrowthorcollapse}, $|P_{\ell+1}| \geq \delta |P_{\leq
			\ell}| > \delta |P_{\ell}|$. Clearly, $s' = (s_0,\dots,s_\ell,s'_{\ell+1},\infty)$ where $\infty
		> s'_{\ell+1} \geq s'_{\ell} = s_{\ell}.$ Thus, the signature vector $s'$ also has increasing
		coordinates and smaller lexicographic value compared to $s$.
		
		\paragraph{At least one layer was collapsed} Let $0 \leq t \leq \ell$ be the index of the last
		layer that was collapsed during the $r$-th iterative step. As a result of the collapse operation
		suppose the layer $P_t$ changed to $P'_t$. Then we know that $|P'_t| < (1-\mu)|P_t|.$
		Indeed, during Step 3 of Algorithm \ref{alg2}, at least a $\mu$-fraction of the edges in $B_t$ are replaced with edges from $I$.
		 Since none
		of the layers with indices less than $t$ were affected during this procedure, $s' =
		(s_0,\dots,s_{t-1}, s'_t,\infty)$ where $s'_t = \lfloor \log_{1/(1-\mu)}
		\frac{|P'_t|}{\delta^{t+1}} \rfloor \leq \lfloor \log_{1/(1-\mu)} \frac{(1-\mu)|P_t|}{\delta^{t+1}}
		\rfloor \leq \lfloor \log_{1/(1-\mu)} \frac{|P_t|}{\delta^{t+1}} \rfloor - 1 = s_t - 1.$ This shows
		that the lexicographic value of the signature vector decreases. That the coordinates of $s'$ are
		non-decreasing follows from Lemma~\ref{lem:expgrowthorcollapse}.
	\end{proof}
	
	Finally, due to the above lemma, any upper bound on the number of possible signature vectors is an upper bound on the number of iterations Algorithm \ref{alg2} will execute; we prove there is such a bound of polynomial size:
	\begin{lemma}
		The number of signature vectors is at most $|\cP|^{O(1/\mu \cdot 1/\delta\cdot \log(1/\delta))}$.
	\end{lemma}
	\begin{proof}
		By Lemma~\ref{lem:expgrowthorcollapse}, $|\cP| \geq P_{\leq \ell} \geq (1+\delta)P_{\leq \ell - 1}
		\geq \dots \geq (1+\delta)^\ell|P_0|.$ This implies that $\ell \leq \log_{1+\delta} |\cP|\leq
		\frac{1}{\delta}\log |\cP|$, where the last inequality is obtained by using Taylor series and
		that $\delta \in [0,1/100]$.

		Now consider the $i$-th coordinate of the signature vector $s_i$. It can be no larger than
		$\log_{1/(1-\mu)} \frac{|\cP|}{\delta^{i + 1}}$.   
		Using the bound on the index $i$ and after some
		manipulations, we get 
		\begin{align*} s_i &\leq \left(\log |\cP|+(i+1)\log
		\frac1\delta\right)\frac{1}{\log \frac{1}{1-\mu}} \\ 
		&\leq \left(\log |\cP|+(\frac{1}{\delta}\log
		|\cP| + 1)\log \frac1\delta\right)\frac{1}{\log \frac{1}{1-\mu}} \\ &= \log |\cP|
		\cdot O\left( \frac{1}{\mu\delta}\log \frac1\delta \right), \end{align*} where the final bound is
		obtained by again expanding using Taylor series around $0$. Thus,  if we let $U= \log |\cP|
		\cdot O\left( \frac{1}{\mu\delta}\log \frac1\delta \right)$ be an upper bound on the number of
		layers and the value of each coordinate of the signature vector, then the sum of coordinates of the
		signature vector is always upper bounded by $U^2$.
		
		Now, as in the simpler algorithm, we apply the bound on
		the number of partitions of an integer. Recall that the number of
		partitions of an integer $N$ can be upper bounded by $e^{O(\sqrt{N})}$~\cite{HardyRam18}. Since each
		signature vector corresponds to some partition of an integer at most $U^2$, we can upper bound
		the total number of signature vectors by $\sum_{i \leq U^2} e^{O(\sqrt{i})}.$

		Now using the bound of $U$, we have that the number of signatures is at most
		$|\cP|^{O(1/\mu \cdot 1/\delta\cdot \log(1/\delta))}$.
		%    signature vector is at most $O_{\mu,\delta}(\log^2 |\cP|).$ Since the coordinates of the signature
		%   vector are non-decreasing, the set of signature vectors can be mapped injectively into the set
		%  of partitions of an integer of size at most $O_{\mu,\delta}(\log^2 |\cP|).$ Finally, the number of
		% partitions of an integer of size $n$ is bounded by $e^{\sqrt{n}}.$
	\end{proof}
	
	Since the number of possible signature vectors is polynomial, the number of iterations Algorithm \ref{alg2} will execute is also polynomial. Furthermore, as the running time of each iteration is also polynomial, this completes the proof of Theorem \ref{thm:main}.

%%% Local Variables:
%%% mode: latex
%%% TeX-master: "main"
%%% End:

\section{Conclusion}
In this paper we have presented new ideas for local search algorithms leading to an improved approximation algorithm for the \problemmacro{restricted max-min fair allocation} 
problem.  The obtained algorithm is also
combinatorial and therefore bypasses the need of solving the exponentially large Configuration-LP.

Apart from further improving the approximation guarantee, we believe that an interesting future direction is to consider our techniques in the more abstract
setting of matchings in hypergraphs. For example, Haxell~\cite{haxell1995condition} proved, using an
\emph{alternating tree} algorithm, a sufficient condition for a bipartite hypergraph to admit a perfect
matching.

  \begin{theorem}[Haxell's Condition]\label{thm:hax}
    Consider an $(r+1)$-uniform bipartite hypergraph $\cH = (\mathcal{P}
    \cup \mathcal{R}, E)$ such that for every edge $e \in E$, $|e \cap
    \cP| = 1$ and $|e \cap \cR|=r$. For $C\subseteq \cP$ let
    $H(E_C)$ denote the size of the smallest set $R \subseteq \cR$
    that hits all the edges in $\cH$ that are incident to some vertex
    in $C$. If for every $C \subseteq \mathcal{P}$, $H(E_C) >
    (2r-1)(|C|-1)$ then there exists a perfect matching in $\cH$.
  \end{theorem}

  Note that Theorem~\ref{thm:hax} generalizes Hall's theorem for
  graphs. However, the proof of the statement does not lead to a
  polynomial time algorithm. In the conference version of this paper
  we had posed the question of whether a constructive analog of
  Theorem~\ref{thm:hax} can be obtained.

  With the techniques presented here, we could prove the following
  weaker statement: there is a constant $C_0 > 0$ for which, given some $0 < \epsilon \leq 1$ and assuming
  $H(E_C) \geq C_0(1/\epsilon)r(|C|-1)$, there exists a polynomial
  time algorithm which assigns one edge $e_p \in E$ for every player
  $p \in \cP$ such that it is possible to choose disjoint subsets
  $\{S_p \subseteq e_p \cap \cR\}_{p\in \cP}$ of size at least
  $(1-\epsilon)r$.

  Recently, the first author obtained such a constructivization
  answering our open question affirmatively~\cite{annamalai2016}. For
  some fixed $\epsilon > 0$ and $r$ he proved that, for
  $(r+1)$-uniform hypergraphs satisfying
  $H(E_C) > (2r-1 + \epsilon)(|C|-1)$ a polynomial time algorithm
  exists for finding the perfect matching guaranteed by
  Theorem~\ref{thm:hax}. However, the running time of this algorithm is
  exponential in both $r$ and $1/\epsilon$. It remains an open problem
  to find such an algorithm whose running time dependence on $r$ is
  polynomial.

% so that the collection of sets
%$\{S_p\}_{p\in\cP}$ are disjoint.

%%% Local Variables:
%%% mode: latex
%%% TeX-master: "main"
%%% End: 

%\Cnote{O: maybe you want to add something here}

%% \input{simplealgo}

%% \input{combalg4}

%% \input{appendix}

%% \newpage

%% \input{combalg2}

%% \newpage

%% \input{thirdtry}
%% \newpage
\bibliographystyle{alpha}

\bibliography{refs}

\end{document}